\newif\ifsingle

\newif\ifFullVersion
\pdfoutput=1
\ifsingle
\documentclass[11pt,draftclsnofoot, onecolumn]{IEEEtran}		
\else		
\documentclass[10pt,final, twocolumn]{IEEEtran}
\fi

\usepackage{times}
\usepackage{amsmath,dsfont}
\usepackage{amssymb,amsthm}
\usepackage{epsfig,verbatim}
\usepackage{subfigure}
\usepackage{setspace}
\usepackage{color}
\usepackage{cite}
\usepackage{epstopdf}
\usepackage{graphics}
\usepackage{accents}
\usepackage{acronym}
\usepackage[bookmarks,colorlinks]{hyperref}
\usepackage{booktabs}
\usepackage{mathtools}
\usepackage{algorithm}
\usepackage{algorithmicx}
\usepackage{algpseudocode}
\usepackage{amsmath}
\usepackage{enumitem}

\newcommand{\myVec}[1]{{\boldsymbol{#1}}}

\newcommand{\mySet}[1]{\mathcal{#1}}

\newtheorem{definition}{Definition}
\newtheorem{theorem}{Theorem}
\newtheorem{corollary}{Corollary}
\newtheorem{proposition}{Proposition}
\newtheorem{lemma}{Lemma}

\ifsingle
\newcommand{\figWidth}{0.65\columnwidth}
\newcommand{\figheight}{3in}

\else
\newcommand{\figWidth}{0.8\columnwidth}
\newcommand{\figheight}{2in}

\fi 

\acrodef{adc}[ADC]{analog-to-digital convertor}
\acrodef{cs}[CS]{compressed sensing}
\acrodef{dtft}[DTFT]{discrete-time Fourier transform}
\acrodef{dnn}[DNN]{deep neural network}
\acrodef{csi}[CSI]{channel state information}
\acrodef{map}[MAP]{maximum a-posteriori probability}
\acrodef{snr}[SNR]{signal-to-noise ratio}
\acrodef{bs}[BS]{base station}
\acrodef{iot}[IOT]{Interent of Things}
\acrodef{mimo}[MIMO]{multiple-input multiple-output}
\acrodef{mse}[MSE]{mean-squared error}
\acrodef{mmse}[MMSE]{minimum \ac{mse}}
\acrodef{pdf}[PDF]{probability density function}
\acrodef{rv}[RV]{random variable}
\acrodef{fec}[FEC]{forward error correction}
\acrodef{dma}[DMA]{dynamic metasurface antenna}
\acrodef{lti}[LTI]{linear time-invariant}
\acrodef{wss}[WSS]{wide-sense stationary}
\acrodef{psd}[PSD]{power spectral density}
\acrodef{ser}[SER]{symbol error rate}
\acrodef{ber}[BER]{bit error rate}
\acrodef{sgd}[SGD]{stochastic gradient descent}
\acrodef{isi}[ISI]{intersymbol interference}
\acrodef{awgn}[AWGN]{additive white Gaussian noise}
\acrodef{ut}[UT]{user terminal}
\acrodef{mmw}[mmWave]{millimeter wave}

\IEEEoverridecommandlockouts

\title{Task-Based Graph Signal Compression}
\author{
	\IEEEauthorblockN{Pei Li, Nir Shlezinger, Haiyang Zhang, Baoyun Wang, and Yonina C. Eldar\\
	}
	\thanks{
		Part of this work has been presented in the IEEE International Conference on Acoustics, Speech, and Signal Processing (ICASSP), 2021 \cite{li2021graph}.	
		This work was supported in part by the National Natural Science Foundation of China under grant No 61971238, by  the European Research Council (ERC) under the European Union's Horizon 2020 research and innovation programme, and by the Israel Science Foundation under grant No. 0100101.
		P. Li and B. Wang are with the School of Communication and Information Engineering, Nanjing University of Posts and Telecommunications, Nanjing, China (e-mail: \{2017010211; bywang\}@njupt.edu.cn)
		N. Shlezinger is with the School of ECE, Ben-Gurion University of the Negev, Beer-Sheva, Israel (e-mail: nirshl@bgu.ac.il)
		H. Zhang and Y. C. Eldar are with the Faculty of Math and CS, Weizmann Institute of Science, Rehovot, Israel (e-mail: \{haiyang.zhang; yonina.eldar\}@weizmann.ac.il).
		}

	\vspace{-1.0cm}
	
}

\begin{document}
	
	\maketitle
	\pagestyle{plain}
	\thispagestyle{plain}
	\begin{abstract}
	Graph signals arise in various applications, ranging from sensor networks to social media data. The high-dimensional nature of these signals implies that they often need to be compressed in order to be stored and transmitted. The common framework for graph signal compression is based on sampling, resulting in a set of continuous-amplitude samples, which in turn have to be quantized into a finite bit representation. In this work we study the joint design of graph signal sampling along with quantization, for graph signal compression. We focus on bandlimited graph signals, and show that the compression problem can be represented as a task-based quantization setup, in which the task is to recover the spectrum of the signal. Based on this equivalence, we propose a joint design of the sampling and recovery mechanisms for a fixed quantization mapping, and present an iterative algorithm for dividing the available bit budget among the discretized samples. Furthermore, we show how the proposed approach can be realized using graph filters combining elements corresponding the neighbouring nodes of the graph, thus facilitating distributed implementation at reduced complexity. Our numerical evaluations on both synthetic and real world data shows that the joint sampling and quantization method yields a compact finite bit representation of high-dimensional graph signals, which allows reconstruction of the original signal with accuracy within a small gap of that achievable with infinite resolution quantizers.

	\end{abstract}
	
	\vspace{-0.2cm}
	\section{Introduction}
	\vspace{-0.1cm}

Rapid development of information and communication technology, 
has lead to a growing need to process and store high-dimensional signals \cite{ref1}. In many families of signals, such as those representing communication networks, social media data, and sensors deployment, the interactions between the elements of the signal obey some graphical structure.  Graph signal processing (GSP) has emerged as a promising technique to deal with such complex signals~\cite{ref2, ref3}.

The high-dimensional nature of graph signals gives rise to the need to compressing them \cite{GSPsurvey}. A leading strategy to compress graph signals is based on sampling their elements  \cite{SamplingTheory, tanaka2020generalized,tanaka2020sampling}. Sampling techniques typically build upon the frequency analysis of graph signals, which is a fundamental tool in GSP, utilized also for  processing such as denoising \cite{FrequencyAnalysis, GSPdenoising} and interpolation \cite{heimowitz2018markov}. Generally speaking, the graph signal values associated with the two end vertices of edges with large weights in the graph tend to be similar\cite{GSPsmooth}. This property, often observed in practice, leads to spectral sparsity, and the resulting signals are referred to as {\em bandlimited} \cite{GSPsparse}. 
Basic graph sampling theory focuses on bandlimited graph signals and relates the spectral support to the number of samples required for representing the signal such that it can be reconstructed from noiseless samples \cite{SamplingTheory}. Nonetheless, in the presence of noise, it is often challenging to determine how
to select which nodes to sample. In general,  sampling set selection is an NP-hard issue, which is often  tackled using greedy approaches \cite{SamplingGreedy}.
Recently, sampling theorems for non-bandlimited graph signals, exploiting sparsity in domains other than the spectral domain, were proposed in \cite{tanaka2020generalized,tanaka2020sampling}. 

A key characteristic of graph signal compression by sampling stems from the fact that it represents the signal by a set of continuous-amplitude samples, where the number of samples is treated as the compression dimension. However, when storing and transmitting graph signals, the level of compression is measured in the number of bits, rather than continuous-amplitude samples, required for representing the signal. This implies that graph signal compression should involve not only sampling, but also {\em quantization} \cite{gray1998quantization}.
To date, existing works on quantization in GSP, e.g., \cite{QuantizationFilter, QuantizationDistributed, QuantizationTimeVarying, QuantizationInterpolation,QuantizationSampling}, focus on applying GSP methods to graph signals with quantized elements \cite{QuantizationFilter, QuantizationDistributed, QuantizationTimeVarying, QuantizationInterpolation} or quantizing sampled graph signals \cite{QuantizationSampling}, rather than designing joint sampling and quantization methods for compressing such signals. This motivates the design of graph signal compression schemes which combine both sampling as well as quantization designed in a joint manner as papers on ADC compression.

In this work, we propose a  compression method for bandlimited graph signals which maps a high-dimensional  signal into a finite-bit representation via sampling and quantization. Our compression method is inspired by the recently proposed task-based quantization framework \cite{HardwareLimited,shlezinger2018asymptotic,Salamtian19task, neuhaus2020taskbased, shlezinger2020taskbased,  shlezinger2019deep}. Task-based quantization studies the acquisition of multivariate continuous-amplitude signals in order to recover some underlying information vector, while operating under an overall bit budget. This is achieved by processing the signal using a dimensionality-reducing combining mapping, carried out in the analog domain, followed by scalar quantization and digital processing, all jointly designed to recover  the task vector \cite{shlezinger2020taskbased}. Our ability to treat graph signal compression as task-based quantization stems from the observation that graph signal sampling can be viewed as an analog combining operation, while the task in reconstructing bandlimited graph signals is the recovery of their spectrum.
%
Unlike the original task-based quantization formulation, which focused on the design of \acp{adc}, and thus considered identical quantization mappings applied to each sampled element, here we consider a compression problem rather than signal acquisition, and thus does not enforce 
this restriction.

Our joint sampling and quantization design employs a graph filter prior to the quantization operation. We consider two types of such graph filters. The first 
is an unconstrained graph filter, which is allowed to carry out arbitrary linear operations on the graph signals during the sampling procedure. For the unconstrained setup, we begin by studying the case in which the number of bits used for representing each sample is fixed, and derive the sampling operator that minimizes the recovery \ac{mse}. Then, we consider an overall bit budget, and optimize the number of bits assigned to each sample. We identify a sufficient condition for which the allocation is optimized by using identical quantizers, i.e., as in conventional task-based quantization~\cite{HardwareLimited}. 

The second considered 
filter 
corresponds to frequency domain graph filters \cite{tanaka2020sampling}. These filters are constrained to represent local computations, such that each node can be combined only with neighbouring nodes. Frequency domain graph filters notably facilitate distributed implementation, compared to unconstrained graph samplers which may require each sample to be produced by observing the entire graph \cite{gama2020graphs}.  For this constrained family, where the sampling operation is modeled by the selection of a subset of the outputs of a frequency domain graph filter, we first  derive the sampling set and the corresponding assigned number of bits. Then we propose an alternating optimization algorithm to design the filter along with the overall compression system.

Our simulation study  considers graph signal compression in three different applications: signals representing sensor networks; meteorological real-world data; and natural images. We consistently demonstrate that the proposed algorithm results in a distortion which is within a minor gap of that achievable without bit constraints, while yielding notably more accurate representations of the graph signal under a given bit budget
compared to separately designed sampling and quantization.

The rest of the paper is organized as follows:
Section~\ref{sec:Model} introduces the graph signal model, and formulates the compression problem.
In Section~\ref{sec:unconstrained}, the compression scheme is presented for generic graph filters, while frequency domain graph filters  are considered in Section~\ref{sec:constrained}.
Section \ref{sec:Sims} details numerical simulations.
Finally, Section~\ref{sec:Conclusions} provides concluding remarks.
Detailed proofs are delegated to the Appendix.

Throughout the paper,
we use lower-case (upper-case) bold characters to denote vectors (matrices).  The transpose, pseudo-inverse and trace of a matrix $\bf{A}$ are respectively denoted as ${\bf{A}}^T$, ${\bf{A}}^{\dag}$ and $\text{Tr}({\bf{A}})$, while $({\bf A})_{i,j}$ is the $(i,j)$th entry of ${\bf A}$, and ${\bf A}_{\mathcal{S}}$ represents a sub-matrix of ${\bf A}$ with rows indexed by  ${\mathcal{S}}$. For a vector $\bf{a}$, ${\bf{a}}_i$ is its $i$-th element, and $\text{diag}$($\bf{a}$) is a diagonal matrix with $\bf{a}$ on its main diagonal. For a scalar $a$, $\lceil a \rceil$ and  $\lfloor a \rfloor$ represent the
round up and round down for $a$, respectively. We use $\mathbb{R}$ for the set of real numbers,  $\mathbb{Z}$ for the integers, and $\boldsymbol{1}_{\mySet{A}}$ is the indicator function for event $\mySet{A}$.

\vspace{-0.2cm}
\section{System Model}
\label{sec:Model}
\vspace{-0.1cm}
 In this section we present the system model for which we derive the joint sampling and quantization  compression method. We first discuss the  graph signal model in Subsection~\ref{subsec:SignalModel}, followed by a brief review of graph sampling and quantization in Subsection~\ref{subsec:Qauntization}, and a formulation of the considered problem in Subsection \ref{subsec:Problem}.

 \vspace{-0.2cm}
\subsection{Bandlimited Graph Signal Model}
\label{subsec:SignalModel}
\vspace{-0.1cm}

We consider an undirected and weighted graph $\mathcal{G}=(\mathcal{V},\mathcal{E},\bf{W})$, in which $\mathcal{V}=\{v_1, v_2,...,v_N\}$ is the
set of nodes and $\mathcal{E}$ is the set of edges. Let ${\bf{W}} \in \mathbb{C}^{N \times N}$ be the adjacency matrix, such that its $(i,j)$th element ${\bf{W}}_{i,j}$ models the similarity/relationship between nodes $i$ and $j$. A graph signal is a function ${\bf f}: \mathcal{V} \to \mathbb{R}^N$ defined on the vertices of $\mathcal{G}$.  We adopt the symmetric normalized Laplacian matrix ${\bf{L}}={\bf{I}}-{\bf{D}}^{-1/2}{\bf{W}}{\bf{D}}^{-1/2}$ as the variation operator, where
${\bf{D}}=\text{diag}\{d_1,d_2,...,d_N\}$
is the degree matrix whose $i$th diagonal element is given by
$d_i=\sum_j {\bf{W}}_{i,j}$.
Since ${\bf{L}}$ is diagonalizable,  there exists an orthogonal matrix ${\bf{U}}$ and a diagonal matrix ${\bf{\Lambda}}$ satisfying ${\bf{L}}={\bf{U}}{\bf{\Lambda}}{\bf{U}}^T$. For a graph signal ${\bf{x}}\in \mathbb{R}^{N}$ defined over $\mathcal{G}$, its graph Fourier transform (GFT) is defined as $\hat{\bf{x}}={\bf{U}}^T{\bf{x}}$. The graph signal  is bandlimited if there exists an integer $K\leq N$ such that its GFT satisfies $\hat{\bf{x}}_k=0$ for all $k\ge K$. The smallest value of $K$ denotes the bandwidth of ${\bf{f}}$.
Bandlimited graph signals can be expressed as ${\bf U}_K{\bf c}$, where ${\bf U}_K$ is the first $K$ columns of $\bf{U}$, while ${\bf c} \in \mathbb{R}^{K}$ is the frequency representation. Here, we consider a noisy observation of a bandlimited graph signal,   given by
\begin{equation}\label{s1}
\begin{aligned}
 {\bf x}={\bf U}_K{\bf c}+{\bf w},
\end{aligned}
\end{equation}
where   ${\bf w}$ is an i.i.d. zero-mean  noise with variance $\sigma_0^2$.

As in the classical factor analysis model  \cite{GraphModel1},
we impose a Gaussian prior on the spectral representation $\bf c$. Specifically, we assume that  $\bf c$ follows a degenerate zero-mean multivariate Gaussian distribution such that ${\bf c} \sim \mathcal{N} (0, \hat {\bf{\Lambda}})$,
where $\hat {\bf{\Lambda}}= \text{diag} \{\sigma_1^2, \sigma_2^2,..., \sigma_K^2\}$, and $\sigma_i^2$ represents the variance for $i$-th spectral representation ${\bf c}_i$ for $i \in \{1,\cdots, K\}$.
Under this model, the   signal ${\bf x}$ obeys a zero-mean multivariate Gaussian distribution with covariance  matrix:
\begin{equation}\label{ex2}
{\bf C}_{\bf x}={\bf U}_K\hat {\bf{\Lambda}}{\bf U}_K^T+\sigma_0^2{\bf I} = {\bf U} \tilde {\bf \Lambda}  {\bf U}^T,
\end{equation}
where
\begin{equation}
\begin{aligned}
\tilde{\bf \Lambda}  = \text{diag}\left\{ {\sigma_1^2+\sigma_0^2}, {\sigma_2^2+\sigma_0^2},...,{\sigma_K^2+\sigma_0^2}, \sigma_0^2,...,\sigma_0^2  \right\}.
\end{aligned}
\end{equation}

The joint Gaussianity of $\bf c$ and $\bf x$ implies that the \ac{mmse}  estimate of the spectrum ${\bf c}$ from the noisy graph signal ${\bf x}$ is given by  $\tilde {\bf c} = {
\bf \Gamma}^*{\bf x}$, where
\begin{equation}\label{ex2p1}
\begin{aligned}
{\bf \Gamma}^*= \hat {\bf{\Lambda}} (\tilde{\bf \Lambda}^{-1})_K {\bf U}^T.
\end{aligned}
\end{equation}
Here, $(\tilde{\bf \Lambda}^{-1})_K$ represents the first $K $ rows of $\tilde{\bf \Lambda}^{-1}$. The resulting estimation error is given by
\begin{equation}\label{ex3p3}
\mathbb{E} \{ \| \tilde{\bf c}- {\bf c}  \|^2 \}= \text{Tr}\left(  \hat {\bf{\Lambda}}
- {\bf U}_K \tilde {\bf{\Lambda}}^2 {\bf U}_K^T{\bf C}_{\bf x}^{-1} \right).
\end{equation}

 \vspace{-0.2cm}
\subsection{Sampling and Qauntization of Graph Signals }
\label{subsec:Qauntization}
\vspace{-0.1cm}
In this work we combine sampling and quantization for compressing bandlimited graph signals. We thus next briefly recall some basics in graph signal sampling and quantization, starting with the  definition of vertex-domain sampling:

\begin{definition}[Vertex-domain sampling \cite{tanaka2020sampling}]
\label{def:Sampling}
A vector ${\bf y} \in \mathbb{R}^P$ is  the vertex-domain sampling of a graph signal ${\bf x} \in \mathbb{R}^N$, if there exists  ${\bf \Psi} \in \mathbb{R}^{P\times N}$ such that ${\bf y} = {\bf \Psi}{\bf x}$ and $P < N$.
\end{definition}
Vertex-domain sampling, abbreviated henceforth as sampling, often restricts
${\bf \Psi}$ to be a submatrix of the $N\times N$ identity matrix \cite{SamplingTheory}. This results in the elements of ${\bf y}$ being also elements of ${\bf x}$. Nonetheless, graph sampling operations  and their corresponding methods for reconstructing ${\bf x}$ from ${\bf y}$  are studied in the literature for various forms of  $ {\bf \Psi}$ \cite{tanaka2020sampling}. Sampling matrices can be generally written as ${\bf \Psi} = {\bf I}_{\mathcal{S}} {\bf F}$, where ${\bf I}_{\mathcal{S}} \in \mathbb{R}^{P \times N}$ is a row-selection matrix, ${\bf F} \in \mathcal{F}\subseteq \mathbb{R}^{N\times N}$ is a  graph filter, and $\mathcal{F}$ is the set of feasible graph filters.

We separately consider two forms of graph filters. The first type is referred to as unconstrained graph filters:
\begin{definition}[Unconstrained graph filter]
\label{def:UncSampling}
An unconstrained graph filter can be any $N\times N$ matrix, i.e., $\mathcal{F}= \mathbb{R}^{N\times N}$.
\end{definition}
For unconstrained graph filters, the sampling matrix ${\bf \Psi}$ can be any matrix in $\mathbb{R}^{P\times N}$. The application of sampling using unconstrained graph filters  requires the entire graph signal to be available for generating each sample, i.e., each element in the sampled $\bf y$ can be affected by every node in the graph. Such sampling procedures can be challenging to implement, particularly when dealing with high-dimensional graph signals. This motivates the restriction to graph sampling procedures which involve only local operations, where each sample is affected only by a subset of neighbouring nodes. Such sampling mechanisms  can be realized by constraining $\bf F$ to represent {\em frequency-domain filtering}, defined next:
\begin{definition}[Frequency-domain graph filter \cite{tanaka2020generalized}]
\label{def:filter}
The set of frequency-domain graph filters defined over a graph $\mathcal{G}$ with Laplacian ${\bf L} = {\bf U}{\bf \Lambda}{\bf U}^T$ is given by
\begin{equation}
    \label{eqn:FDGfilter}
    \mathcal{F} = \{{\bf F} = {\bf U}{F}({\bf \Lambda}) {\bf U}^T | {F}({\bf \Lambda}) \text{ is diagonal}\}.
\end{equation}
%
\end{definition}
Restricting the sampling matrix ${\bf \Psi}$ to  frequency-domain  filtering results in the sampled vector ${\bf y}$ taking the form
\begin{equation}
\label{eqn:FDsample}
{\bf y} = {\bf I}_{\mathcal{S}}  {\bf U}{F}({\bf \Lambda}) \hat{\bf x}.
\end{equation}
The representation of the graph sampling operation via \eqref{eqn:FDsample} notably facilitates distributed operation compared with sampling using unconstrained graph filters. This follows since  \eqref{eqn:FDsample} can be approximated by interpolation \cite{2011Wavelets}, resulting in the frequency-domain graph filter being expressed as a matrix polynomial function with respect to 
 $\bf L$. In particular, a frequency domain graph filter restricted to combining elements corresponding the neighbouring nodes of the graph  implies that ${{F}}({\bf \Lambda})$ should be a polynomial function with respect to ${\bf \Lambda}$, i.e., ${{F}}({\bf \Lambda}) = \sum_{i = 0}^{K_0} \beta_i {\bf \Lambda}^i$, for some coefficients $\{\beta_i\}$, where $K_0$ is the length of the polynomial \cite{gama2020graphs}.


Next, we discuss the  considered quantization rule. While in general, quantization encapsulates a broad range of continuous-to-discrete mappings \cite{gray1998quantization}, here we focus on conventional uniform scalar quantizaiton, defined as follows:
\begin{definition}[Uniform quantization]
\label{def:Uquant}
A uniform quantizer with resolution $M$, support $\gamma$, and step size $\delta = \frac{2\gamma}{M}$, is the mapping 
\begin{equation*}
	Q_M\left( x \right) \triangleq \begin{cases}
		\delta \left( \left\lfloor \frac{x}{\delta} \right\rfloor + \frac{1}{2} \right), & \mathrm{if~} \vert x \vert < \gamma,\\
		\mathrm{sign} \left( x \right) \left( \gamma - \frac{\delta}{2} \right), & \mathrm{else}.
	\end{cases}
\end{equation*}
\end{definition}
The non-linear nature of quantization mappings results in a complex model for the distortion induced in this procedure, i.e., the term $x - Q_M\left( x \right)$. In our analysis we model the quantization operation as {\em non-subtractive dithered quantization}, obtained by adding noise prior to uniform quantization \cite{gray1993dithered}. The main motivation for this model is that it results in the distortion  being uncorrelated with the input signal when the quantizer is not overloaded, i.e., $|x| \leq \gamma$. This resulting model, which notably facilitates the design and analysis of quantization systems, is also a good approximation of the distortion model for various  quantizer input distributions without dithering as analyzed in \cite{widrow1996statistical} as demonstrated in \cite{HardwareLimited}.

\vspace{-0.2cm}
\subsection{Problem Formulation}
\label{subsec:Problem}
\vspace{-0.1cm}
We consider the problem of compressing a noisy  graph signal ${\bf x}$ into a digital representation comprised of $\log_2 M$ bits. This compressed version should preserve the information required to reconstruct the bandlimited graph signal ${\bf U}_K{\bf c}$. We assume that the structure of the graph signal and its spectral support are known, i.e., ${\bf U}_K$ is given.

Since the volume of the representation is limited in its overall number of bits, we study compression mechanisms consisting of both sampling and quantization, 
as defined in Subsection~\ref{subsec:Qauntization}. In such a system, the graph signal is first sampled by applying the $P\times N$ matrix ${\bf \Psi}$, with $P< N$, and the entries of the sampled  ${\bf y} = {\bf \Psi} {\bf x}$ are then discretized by the uniform quantizers $\{Q_{M_i}(\cdot)\}_{i=1}^{P}$, resulting in the finite-bit representation $Q({\bf  y}) = [Q_{M_1}(y_1),\ldots, Q_{M_P}(y_P)]^T$. The overall bit constraint implies that $\sum_{i=1}^{P} \log_2 M_i \leq \log_2 M$.

As quantization inherently results in lossy compression \cite[Ch. 23]{polyanskiy2014lecture}, we do not aim to perfectly recover  ${\bf x}$, and use the \ac{mse} as our design objective.
Consequently, the joint design of the sampling and quantization mappings can be formulated as recovering the digital representation from which the bandlimited portion of ${\bf x}$, i.e., ${\bf U}_K {\bf c}$, can be reconstructed most accurately. This is mathematically formulated as \begin{equation}\label{problem1}
\tag{P1}
\begin{aligned}
\mathop {\min }\limits_{{\bf \Psi}   ,Q\left(  \bullet  \right)} & \mathbb{E}\left\{ {{{\left\| {\bf U}_K {\bf c} -  \mathbb{E}\left\{ {\bf U}_K {\bf c}| Q({{\bf \Psi} {\bf{x}}})\right\}  \right\|}^2}} \right\},\\
\text{s.t.}\ &\sum_{i=1}^P \log_2 M_i \le \log_2 M, \quad M_i \in \mathbb{Z}^+, i \in \mySet{P},
\end{aligned}
\end{equation}
where $\mySet{P}\triangleq  \{1,2,\ldots,P\}$.
In the following sections we jointly design the resulting compression system based on  \eqref{problem1}.

	\vspace{-0.2cm}
	\section{Joint Sampling and Qauntization Method with Unconstrained Graph Filters}
	\label{sec:unconstrained}
	\vspace{-0.1cm}
	In this section we derive joint sampling and quantization methods for graph signal compression with unconstrained graph filters. To that aim, we first present in Subsection~\ref{subsec:Alternative} an alternative problem formulation, obtained by introducing some relaxations to  \eqref{problem1}. 
	Then, we present the MSE minimizing system design in Subsection~\ref{subsec:MSE minimizng}, and provide a greedy-based algorithm to tackle it in Subsection~\ref{subsec:Algorithm}.

\vspace{-0.2cm}
\subsection{Alternative Optimization Problem}
\label{subsec:Alternative}
\vspace{-0.1cm}
The problem formulation \eqref{problem1} characterizes the graph signal compression setup as designing the sampling matrix $\bf \Psi$ and the scalar quantizers $Q(\bullet)$.
The resulting setup bears much similarity to task-based quantization, proposed as a framework for designing \acp{adc} to extract information from an acquired analog signal \cite{HardwareLimited,shlezinger2018asymptotic,Salamtian19task,neuhaus2020taskbased}. To exploit task-based quantization in graph signal compression, we formulate an alternative problem based on \eqref{problem1}, which can be tackled using these techniques.

First, we note that since the unitary matrix ${\bf U}$ and its submatrix ${\bf U}_K$ are known, 
recovering  the \ac{mmse} estimate of ${\bf x}$, as formulated in \eqref{problem1}, is equivalent to the corresponding estimation of ${\bf c}$. We thus henceforth refer to ${\bf c}$ as the {\em task vector}.  Furthermore, recalling the definition of $\tilde{\bf c} = \mathbb{E}\{{\bf c} | {\bf x}\} = {
\bf \Gamma}^*{\bf x}$, then by the orthogonality principle, designing $Q(\bullet)$ and ${\bf \Psi}$ to minimize the \ac{mse} w.r.t.  ${\bf c}$ is the same as designing them to minimize  the \ac{mse} w.r.t.  $\tilde{\bf c}$. The objective in  \eqref{problem1} thus becomes
\vspace{-0.1cm}
\begin{equation}
    \label{eqn:Objective2}
    \mathop {\min }\limits_{{\bf \Psi}   ,Q\left(  \bullet  \right)}  \mathbb{E}\left\{ {{{\left\| \tilde{\bf{c}} -  \mathbb{E}\left\{\tilde{\bf{c}} | Q({{\bf \Psi} {\bf{x}}})\right\}  \right\|}^2}} \right\}.
    \vspace{-0.1cm}
\end{equation}

Next, we relax  \eqref{eqn:Objective2} by focusing on linear recovery. Our motivation for considering linear schemes stems from their analytical tractability, and since they are commonly used for reconstructing sampled graph signals \cite{tanaka2020sampling}.  The compression system is designed such that the desired  vector ${\bf c}$ can be recovered from $Q({{\bf \Psi} {\bf{x}}})$ using a filter ${\bf \Phi}\in \mathbb{R}^{K\times P}$, i.e., the recovered ${\bf c}$ is given by $\hat{\bf c} = {\bf \Phi}Q({\bf y})$, where ${\bf y} = {{\bf \Psi} {\bf{x}}}$. This compression and recovery system is illustrated in Fig.~\ref{fig1}.

\begin{figure*}
\centering
\includegraphics[width=0.8\linewidth]{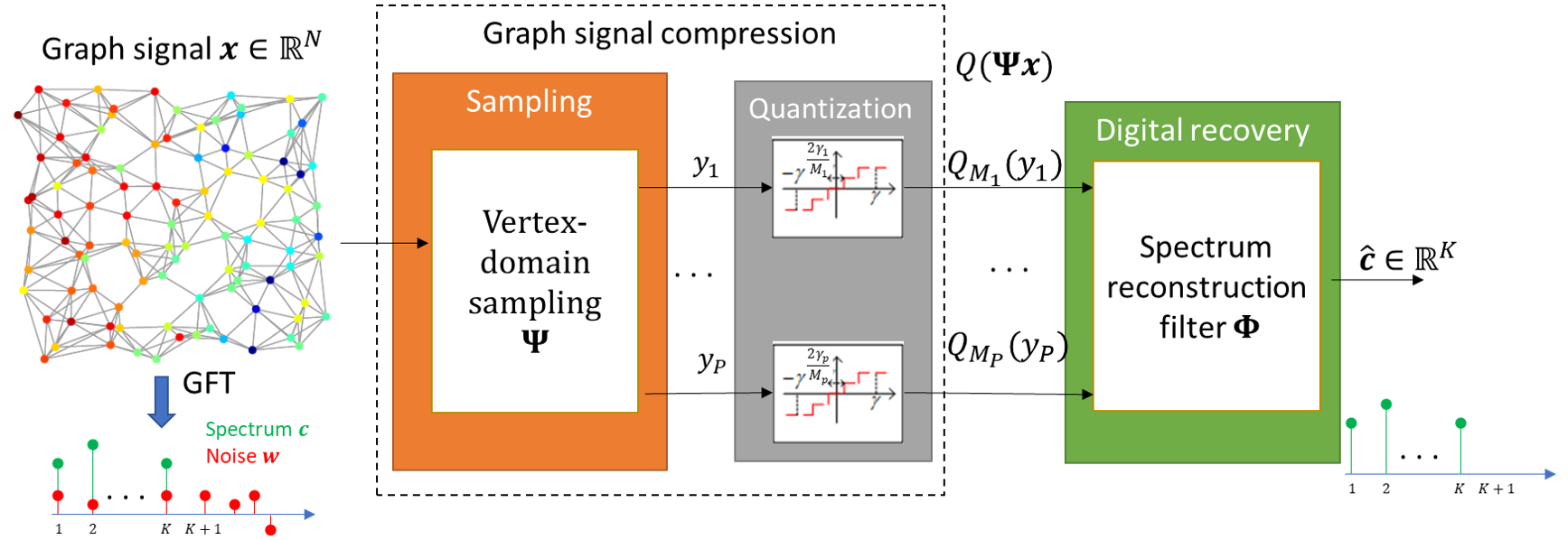}
\vspace{-0.2cm}
\caption{Graph signal compression by sampling and quantization.}\label{fig1}
\vspace{-0.4cm}
\end{figure*}

Quantizers are typically designed to operate within their dynamic range \cite{gray1998quantization}. Therefore, following \cite{HardwareLimited}, we guarantee that the probability of overloading the quantizers is sufficiently small, i.e., that  $\text{Pr}(|({\bf \Psi}{\bf x})_i|>\gamma_i) \approx 0$, by fixing the support of the $i$th quantizer to $\gamma_i^2 = \eta^2  \mathbb{E}\left\{({{\bf \Psi} {\bf{x}}})_i^2\right\}$ for each $i\in\mySet{P}$. For instance, setting $\eta = 2$ results in overloading probability $<5\%$. When the overloading probability vanishes, then the output of the dithered quantizers can be written as \cite{gray1993dithered}
\vspace{-0.1cm}
\begin{equation}\label{cons1}
Q\left( {{\bf \Psi} {\bf{x}}} \right)={{\bf \Psi} {\bf{x}}} +{\bf e}_Q,
\vspace{-0.1cm}
\end{equation}
where the quantization error vector ${\bf e}_Q$ has i.i.d. zero-mean entries of variance  $\delta_i^2/6$ (with $\delta_i = 2\gamma_i / M_i)$, i.e.,
\vspace{-0.1cm}
\begin{equation}\label{cons2}
\begin{aligned}
{\bf G} \triangleq  \mathbb{E}\left\{{\bf e}_Q{\bf e}_Q^T  \right\}=\text{diag}\left\{\frac{2\gamma_1^2}{3 M_1^2},\frac{2\gamma_2^2}{3 M_2^2},\cdots,\frac{2\gamma_P^2}{3M_P^2}\right\}.
\end{aligned}\vspace{-0.1cm}
\end{equation}
While our setting of $\gamma_i^2$ achieves a small yet non-zero overloading probability,  we design the compression mechanism by utilizing the model in \eqref{cons1}, which rigorously holds when the overloading probability is zero, as proposed in \cite{HardwareLimited,shlezinger2018asymptotic}.

To summarize, the alternative optimization problem based on which we design our scheme is given by
\begin{equation}\label{problem2}
\tag{P2}
\begin{aligned}
\mathop {\min }\limits_{{\bf \Psi} ,{\bf \Phi} , \{M_i\}} & \mathbb{E}\left\{ {{{\left\| {{{\bf \Gamma}^*{\bf x}} - {\bf \Phi} ({{\bf \Psi} {\bf{x}}}+ {\bf e}_Q)} \right\|}^2}} \right\},\\
\text{s.t.}\ &\sum_{i=1}^P \log_2 M_i \le \log_2 M,\quad M_i \in \mathbb{Z}^+, i \in \mySet{P}.
\end{aligned}
\end{equation}
The optimization problem \eqref{problem2} can be treated as a task-based quantization problem. In task-based quantization, an analog filter is designed along with the overall acquisition system in light of the system task \cite{HardwareLimited}. Here, the sampling operation implements the pre-quantization combining, and the system task is to recover the vector ${\bf c}$ (via its \ac{mmse} estimate). However, as task-based quantization focuses on the design of \acp{adc} operating in a serial manner, it is restricted to utilizing identical quantizers, i.e., each quantizer uses the same number of bits, while our graph signal compression problem does not share this restriction. Therefore, our design based on \eqref{problem2}, given in the following subsection, combines task-based quantization methods with dedicated bit allocation techniques.

\vspace{-0.2cm}
\subsection{Compression System Design}
\label{subsec:MSE minimizng}
\vspace{-0.1cm}	
In this section, we design a joint graph signal sampling and quantization scheme based on the identified relationship between such setups and task-based quantization in   \eqref{problem2}. Directly solving \eqref{problem2} is difficult due to the coupling between its optimization variables, combined with the non-linear relationship between these parameters and the statistical model of the quantization error, observed in   \eqref{cons2}. Therefore, in the following, we  first design the sampling matrix based on \eqref{problem2} for a fixed bit allocation, i.e.,  when $\{M_i\}$ is given.


As a preliminary step, we note that for any given sampling matrix ${\bf \Psi}$ and bit allocation $\{M_i\}$, the MSE minimizing linear recovery matrix ${\bf \Phi}$ is given in the following lemma:
\begin{lemma}
\label{lem:Digital}
For any fixed ${\bf \Psi}$ and $\{M_i\}$, the objective in \eqref{problem2} is minimized by setting
\begin{equation}\label{lem13p1}
\begin{aligned}
{\bf \Phi}^*= {\bf \Gamma}^*{\bf C}_{\bf x} {\bf \Psi}^T\left(  {\bf \Psi} {\bf C}_{\bf x} {\bf \Psi} ^T + {\bf G} \right)^{-1}.
\end{aligned}
\end{equation}
\end{lemma}

\begin{IEEEproof}
The lemma is obtained following \cite[App. B]{HardwareLimited}.
\end{IEEEproof}

\smallskip
Next, we design the sampling matrix ${\bf \Psi}$. Here, we recall the similarity between our sampling matrix and analog combiner design in task-based quantization setups. In particular, the MSE minimizing analog combiner for task-based quantization is given in \cite[Thm. 1]{HardwareLimited} for an identical bit allocation $\tilde{M} = \lfloor M^{1/P}\rfloor $.
For such setups, the MSE is minimized for a setting of the form ${\bf \Psi} = {\bf U} {\bf \Xi} {\bf V}^T$, where ${\bf V}\in \mathbb{R}^{N\times N}$ is the right singular vectors matrix of ${\bf \Gamma}^*{\bf C}_{\bf x}^{1/2} $, ${\bf \Xi}\in \mathbb{R}^{P\times N}$ is diagonal with non-negative diagonal entries, and ${\bf U}\in \mathbb{R}^{P\times P}$ is a unitary matrix designed to balance the inputs to the identical quantizers.
Since in our setup, the quantizers are not restricted to be identical, we cannot adopt the results derived in \cite{HardwareLimited}. Nonetheless, for a given bit allocation $\{M_i\}$ sorted in a descending order, we can characterize the \ac{mse} minimizing sampling matrix, as stated in the following proposition:
\begin{proposition}
\label{lem:MSE}
For a bit allocation $\{M_i\}$ with $M_i\ge M_{i+1}$,  the MSE minimizing unconstrained sampling matrix satisfies
\vspace{-0.1cm}
\begin{equation}
{\bf \Psi}^*= {{\bf{U}}_\Psi }{\bf \Xi}_\Psi   { \tilde{\bf \Lambda}}^{-1/2}  {{\bf{ U}}} ^T,
\vspace{-0.1cm}
\end{equation}
where ${{\bf{U}}_\Psi }$ is a unitary matrix satisfying
\vspace{-0.1cm}
\begin{equation}
({{\bf{U}}_\Psi } {\bf \Xi}_\Psi{\bf \Xi}_\Psi^T{\bf{U}}_{\Psi}^T)_{i,i} = \frac{3M_i^2}{2\eta^2},
\vspace{-0.1cm}
\end{equation}
and ${\bf \Xi}_\Psi \in \mathbb{R}^{P\times N} $ is a diagonal matrix with non-negative entries. 
Denoting ${\alpha _i} = ({\bf \Xi}_\Psi)_{i,i}^2, i\in \mySet{P}$, the elements of ${\bf \Xi}_\Psi$ are obtained as the solution to the following problem:
\vspace{-0.1cm}
\begin{equation}\label{MSESam2}
\begin{aligned}
\{\alpha_i\}_{i=1}^{P} =&\mathop {\arg \min }\limits_{\{\alpha'_i\}}~  \sum_{i=1}^P \frac{{{\lambda _{{\bf{\Gamma }},i}^2}}}{{{\alpha'_i} + 1}},\\
s.t.~ &  {\frac{{3}}{2\eta^2}}[M_1^2,\ldots M_P^2]^T\prec  [a'_1, \ldots a'_P]^T,
\end{aligned}
\vspace{-0.1cm}
\end{equation}
where $\prec$ denotes majorization, i.e., $\myVec{a} \prec \myVec{b}$ implies that $\myVec{a}$ is majorized by $\myVec{b}$ \cite{palomar2007mimo}.  In \eqref{MSESam2}, ${\lambda _{{\bf{\Gamma }},i}}$ is the $i$th largest singular value of ${\bf \Gamma}^*{\bf C}_{\bf x}^{1/2} $.  The resulting excess \ac{mse} is given by:
\vspace{-0.1cm}
\begin{equation}\label{MSELambda}
\begin{aligned}
\mathbb{E} \{ \| \hat{\bf c}- \tilde{\bf c}  \|^2 \} = \sum_{i=1}^P \frac{{{\alpha _i}{\lambda _{{\bf{\Gamma }},i}^2}}}{{{\alpha _i} + 1}}.
\end{aligned}
\end{equation}
\end{proposition}
\begin{IEEEproof}
The proof is given in Appendix \ref{app:proof1}.
\end{IEEEproof}

 For the special case of identical bit allocation, i.e., $M_i=M_j$ for each $i,j \in \mySet{P}$, it can be shown that Proposition~\ref{lem:MSE} coincides with \cite[Thm. 1]{HardwareLimited}.
We note that \eqref{MSESam2} is a convex problem, which can be solved using existing convex optimization toolboxes such as  CVX. The majorization constraint can not be addressed with CVX directly, So we split it into equivalent multiple linear constraints, which are all convex and can be solved using CVX.
Consequently, while Proposition \ref{lem:MSE} does not give the sampling matrix ${\bf \Psi}^*$ in closed form, it can be numerically computed with an affordable computational effort. Nonetheless, identifying the bit allocation $\{M_i\}$ which minimizes the \ac{mse} based on Proposition~\ref{lem:MSE} is a challenging task due to the discrete nature of the bit assignment, motivating the greedy optimization algorithm detailed in the following section. However, if one is allowed to assign non-integer bit values, the resulting compression system and its corresponding \ac{mse} can be obtained explicitly via the following theorem:
%
%
\begin{theorem}
\label{thm:MSE}
When $\{M_i\}$ are not limited to be integer, the optimal unconstrained sampling operator is ${\bf \Psi}^*=  {\bf U}_K^T$. The \ac{mse} minimizing bit allocations satisfies
\begin{align}
 &M_i^2 \!= \! \begin{cases}
  \eta^2\frac{-2\beta^*+\lambda_{{\bf{\Gamma }},i}^2+ \sqrt{\lambda_{{\bf{\Gamma }},i}^4 - 4 \beta^* \lambda_{{\bf{\Gamma }},i}^2}}{3\beta^*} &   {  \beta^* \le    \frac{\lambda_{{\bf{\Gamma }},i}^2}{4} ,}\\
1 & {\text{otherwise}},
\end{cases}
\label{algorithms1}
\end{align}
where  the hyperparameter  $\beta^*$ is set such that $\sum_{i=1}^P \log_2 M_i = \log_2M$.
The resulting excess \ac{mse} is given by:
\begin{equation}\label{AchMSE}
\begin{aligned}
\mathbb{E} \{ \| \hat{\bf c}- \tilde{\bf c}  \|^2 \} = \sum_{i=1}^{K}\lambda_{\tilde \Gamma,i}^2- \sum_{i=1}^{\min(K,P)} \frac{{{3M_i^2}}\lambda_{\tilde \Gamma,i}^2}{{3M_i^2}+{2\eta^2}}.
\end{aligned}
\end{equation}
\end{theorem}

\begin{IEEEproof}
The proof is given in Appendix~\ref{app:proof2}.
\end{IEEEproof}

\smallskip
While Theorem \ref{thm:MSE} considers a hypothetical system which can quantize using non-integer number of bits, it reveals an important challenge arising from the incorporation of quantization  compared to conventional graph sampling. When one can quantize with arbitrary non-integer levels, the resulting compression problem reduces to a conventional graph sampling (without quantization) setup, as shown in Appendix~\ref{app:proof2}. For such cases, the sampling matrix ${\bf \Psi}^*$ specializes into conventional frequency domain sampling of bandlimited signals, which settles with the corresponding results in \cite{tanaka2020sampling}. However, since quantization is inherently restricted to discrete levels, the \ac{mse} minimizing graph sampling matrix does not necessarily reduce to frequency domain sampling. In such cases, one has to utilize Proposition~\ref{lem:MSE} combined with dedicated schemes for just rounded or  optimizing the bit allocation, as proposed next.


\vspace{-0.2cm}
\subsection{Greedy Optimization Algorithm Design}
\label{subsec:Algorithm}
\vspace{-0.1cm}	
The MSE characterized Theorem \ref{thm:MSE} is generally not achievable  for graph signal compression schemes since it ignores the fact that the bit assignment must take integer values.
Nonetheless, the closed-form \ac{mse} expression \eqref{AchMSE} can be used to facilitate the setting of the (discrete) bit allocation $\{M_i\}$ using greedy optimization. 
In the following we first introduce the greedy algorithm for setting the bit budget, after which we identify a sufficient and necessary conditions for which the bit assignment is designed by using identical quantizers, and identify the number of quantizers $P$ which minimizes the \ac{mse} without setting any quantizer to be inactive.

\subsubsection{Greedy Bit Assignment}
The proposed greedy algorithm starts by setting the minimal bit allocation for all the quantizers, i.e., $M_i^{(0)}=1$ for each $i \in \mySet{P}$. Then, it increments the number of levels for the quantizer which contributes most to the objective \eqref{AchMSE}. While this procedure does not impose the  constraint $M_i \ge M_{i+1}$ for each $i \in \{1,2,...,P-1\}$, it is implicitly maintained due to the descending order of  $\{\lambda_{ {\bf \Gamma},i}\}$.
Specifically, the gradient of the objective \eqref{AchMSE} w.r.t. $M_i$ is $g(M_i) := -\frac{\partial}{\partial M_i} \sum_{i=1}^{\min(K,P)} \frac{3M_i^2\lambda_{\tilde \Gamma,i}^2}{3M_i^2+{2\eta^2}}$, which equals
\vspace{-0.1cm}
\begin{align*}
 g_i(M_i)
 &= \! \begin{cases}
{ -\frac{12M_i\eta^2\lambda_{{\bf \Gamma},i}^2}{(3M_i^2 \!+ \!2\eta^2)^2}  }
&{  i\le K,  }\\
0 & {\text{otherwise}}.
\end{cases}
\vspace{-0.1cm}
\end{align*}
At the $k$th iteration, the gradient vector ${\bf g}^{(k)}$ is thus given by
\begin{equation}\label{algorithms2}
{\bf g}^{(k)} = [  g_1(M_1^{(k)}),  g_2(M_2^{(k)}),..., g_P(M_P^{(k)}) ]^T.
\end{equation}

The elements in \eqref{algorithms2} respectively represent the distortion decreasing efficiency of the extra level we assign to each quantizer at the $k$th iteration.
The smallest entry of \eqref{algorithms2}   determines which quantizer is assigned an additional  level, repeating until the  budget $M$ is exhausted. The resulting greedy  allocation algorithm is summarized as Algorithm~\ref{alg:greedy}. 
Once the bit assignment is obtained, the sampling operator and its corresponding \ac{mse} are obtained using Proposition~\ref{lem:MSE}, while the digital reconstruction filter is computed via Lemma~\ref{lem:Digital}.

\begin{algorithm}
\caption{Greedy bit allocation algorithm}
\label{alg:greedy}
\hspace*{0.02in} {\bf Input:}
$\{\lambda_{\tilde \Gamma,i}\}$\\
\hspace*{0.02in} {\bf Output:}
Bit allocation $\{M_i\}$.\\
\hspace*{0.02in} {\bf Initialize:}
$k=1$ and set $M_i^{(0)}=1$ for $\forall i \in \mySet{P}$.
\begin{algorithmic}[1]
\While {$\sum_{i=1}^P \log_2M_i^{(k)} \le \log_2M$}
                    \State Compute ${\bf g}^{(k)}$ via \eqref{algorithms2};
                    \State Update $M_i^{(k+1)}=M_i^{(k)}$ + $\boldsymbol{1}_{i=\arg \min_j g_j(M_j^{(k)})}$;
                    \State  $k = k+1$;
                \EndWhile
\State \Return $\{M_i^{(k-1)}\}$
\end{algorithmic}
\end{algorithm}

\subsubsection{Analysis of Algorithm~\ref{alg:greedy}}
In the  method summarized as Algorithm~\ref{alg:greedy}, we use a greedy approach based on an \ac{mse} measure achievable using real-valued assignments to optimize the discrete bit allocation. Therefore, to assess the validity of this approach, we first numerically compare the ability of Algorithm~\ref{alg:greedy} to approach the excess \ac{mse} in \eqref{AchMSE} using discrete bit assignments. Then, we characterize a sufficient condition for which Algorithm~\ref{alg:greedy} yields an identical bit allocation. Finally, since the number of quantized samples $P$ is inherently a design parameter of the system, we derive the setting of $P$ which optimizes the \ac{mse} achievable using Algorithm~\ref{alg:greedy}.

{\bf Empirical Evaluation}:  Fig.~\ref{figgreedy} compares the \ac{mse} achieved with the discrete bit assignment computed using Algorithm~\ref{alg:greedy} to the \ac{mse} evaluated  via \eqref{AchMSE} for $30$ i.i.d. examples with $P=K=10$. As the purpose of this study is to evaluate the ability of Algorithm~\ref{alg:greedy} to approach \eqref{AchMSE}, we consider a toy example where the descending sequence $\{\lambda_{\tilde \Gamma,i}\}$ is randomly generated with normal Gaussian distribution;  
more realistic experiments corresponding to different forms of graph signals are reported in our simulation study in Section~\ref{sec:Sims}.

\begin{figure}
\centering
\vspace{-0.4cm}
\subfigure[$\log_2M = 10$]{
\includegraphics[width=1.6in]{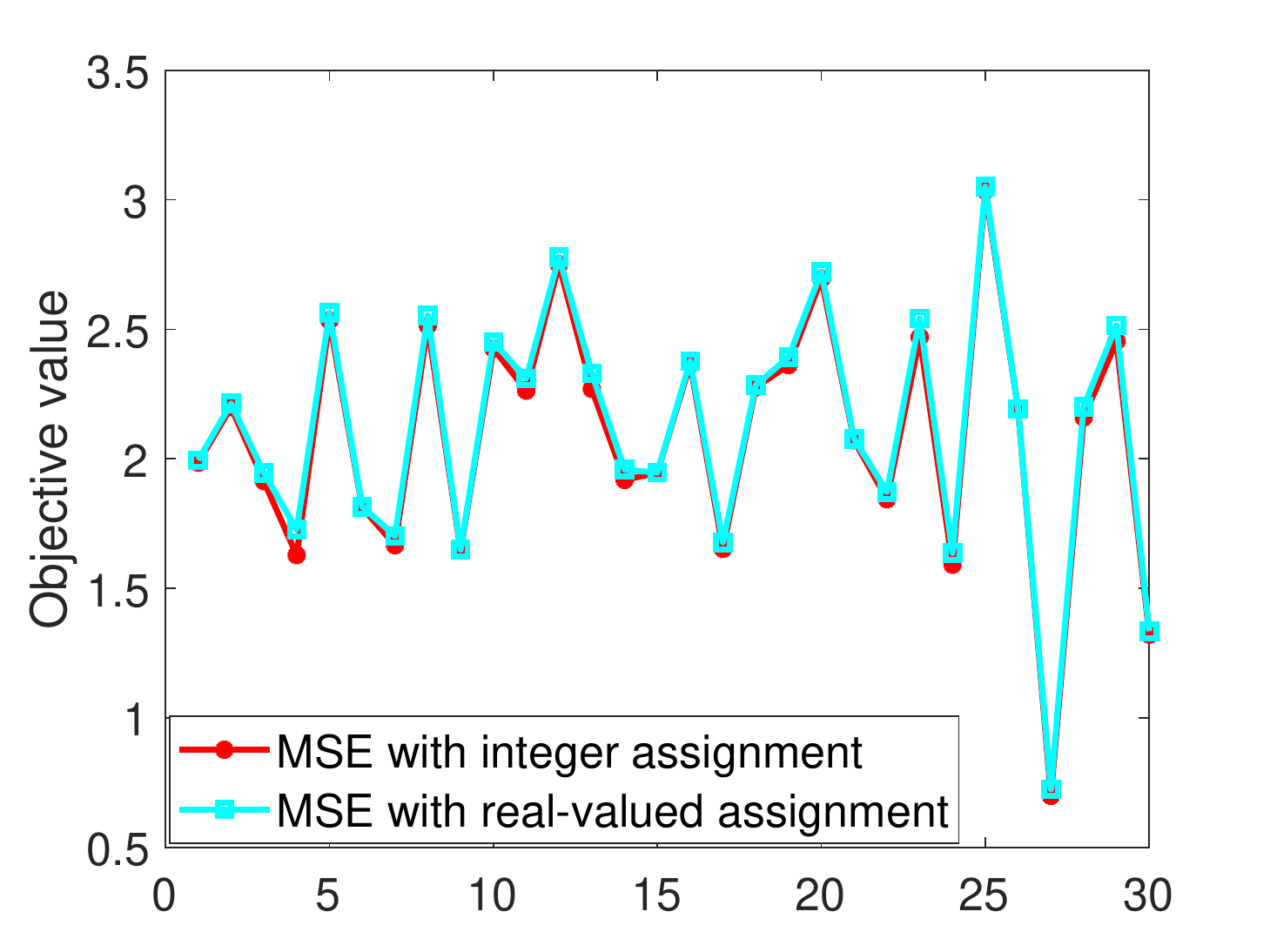}
\label{figgreedyA}
}
\subfigure[$\log_2M = 20$]{
\includegraphics[width=1.6in]{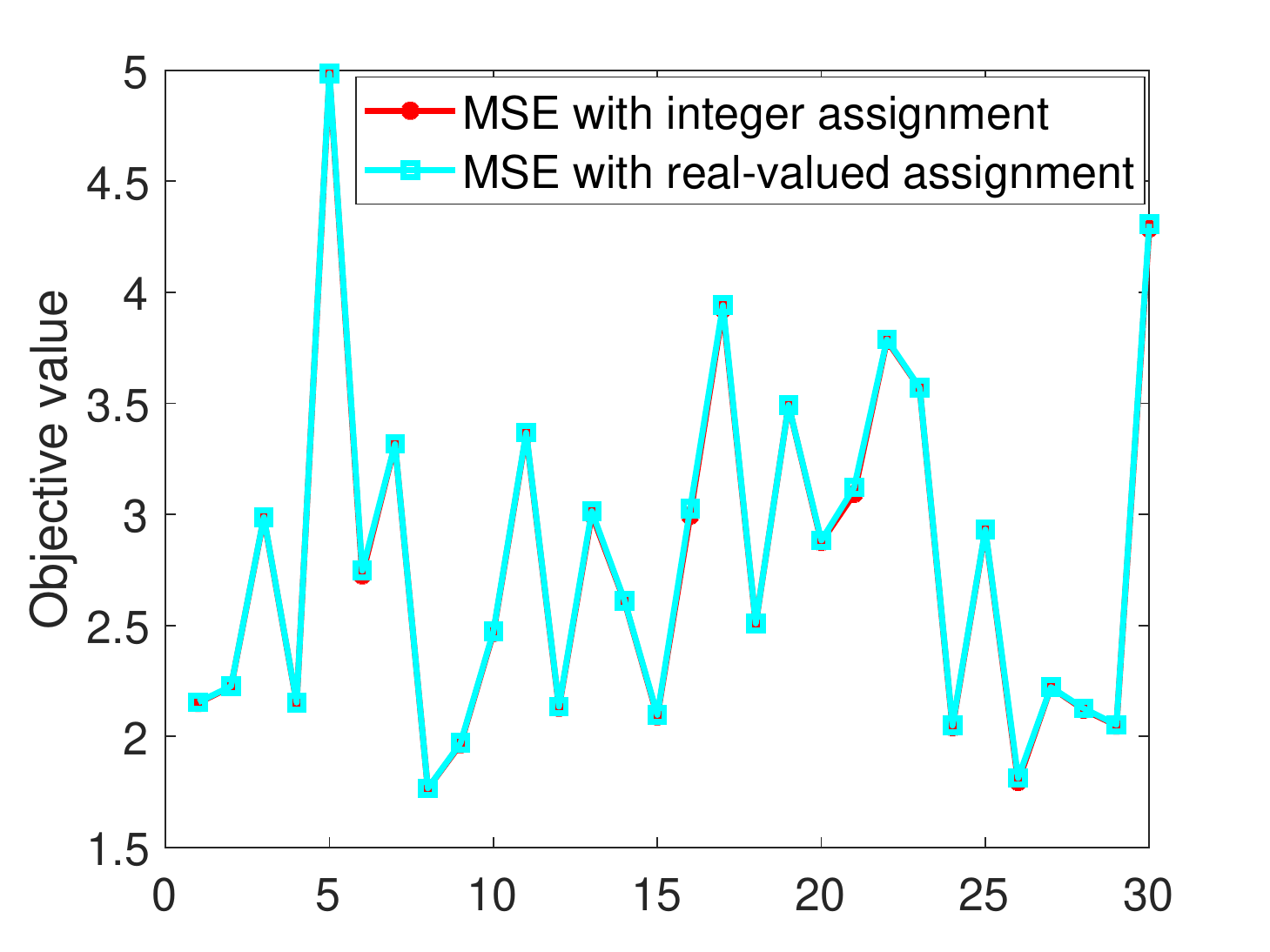}
\label{figgreedyB}
}
\caption{$30$ examples with randomly generated ${\lambda_{ \Gamma,i}}$. 
}
\label{figgreedy}
\end{figure}

Observing Fig.~\ref{figgreedy}, we note that when the overall bit budget is tight, i.e., $\log_2 M = 10$ bits as in Fig.~\ref{figgreedyA}, 
there is a small gap between the \ac{mse} achieved using the discrete assignment and that which can reached given the ability to assign arbitrary   quantization levels. However, as the bit budget increases to $\log_2 M = 20$, which is still a relatively tight budget (allowing merely 2 bits per quantizer for a standard identical bit assignment), we observe in Fig.~\ref{figgreedyB} that the gap becomes negligible. This study indicates that although Algorithm~\ref{alg:greedy} is designed based on an \ac{mse} measure typically not achievable with integer bit assignments, it allows to approach it to within a small gap, which effectively vanishes as the bit budget grows.




{\bf Identical Bit Assignment}: As discussed in Subsection~\ref{subsec:MSE minimizng}, while our derivation relies on representing the joint sampling and quantization of bandlimited graph signals as a task-based quantization setup, the resulting formulation differs from that studied in the context of \ac{adc} design in \cite{HardwareLimited}. One of the key differences between the graph signal compression task considered here and the design of hybrid analog/digital acquisition studied in \cite{HardwareLimited} stems from the additional degree of freedom in the ability to allocate different bit assignments between the quantizers. As the restriction to utilize identical quantizers may facilitate the design of the compression system, we next identify a sufficient and necessary condition for which Algorithm~\ref{alg:greedy} yields an identical bit allocation:

\begin{corollary}
\label{co:eq}
Algorithm~\ref{alg:greedy} yields an identical bit assignment $M_i \equiv M_a := \lfloor M^{1/P} \rfloor $ for each $i=1,\ldots, P$ where $P\leq K$ when the following conditions are satisfied:
\begin{subequations}
\label{eqn:coeq}
\begin{align}
\label{eqn:coeq1}
(M_a)^P \le M < (M_a)^P + (M_a)^{P-1}, \\
\label{eqn:coeq2}
g_1(M_a) < g_P(M_a - 1).
\end{align}
\end{subequations}
\end{corollary}


\begin{IEEEproof}
The proof is given in Appendix~\ref{app:co1}.
\end{IEEEproof}

\smallskip
Corollary~\ref{co:eq} identifies conditions for which, when satisfied, one can simply utilize conventional identical bit allocation rather than go through the greedy optimization in Algorithm~\ref{alg:greedy}. Nonetheless, the conditions \eqref{eqn:coeq} are not that commonly satisfied by graph signals.
The condition \eqref{eqn:coeq2} hints that the $\lambda_{\Gamma, P}$  should be close to  $\lambda_{\Gamma, 1}$. However, the smoothness of the graph signal determines that the weight of the low-frequency component is usually much greater than the weight of the high-frequency component, i.e., $\sigma_1 \gg \sigma_P$, which potentially contradicts \eqref{eqn:coeq2}.
On the other hand, when \eqref{eqn:coeq1} is not satisfied, the greedy-based algorithm we proposed can also improve the performance by using the redundant bits.



{\bf Number of Quantized Samples}: Algorithm~\ref{alg:greedy} assigns the overall bit budget among  $P$ quantized samples, where $P$ is fixed. Nonetheless, it may yield an assignment $M_i = 1$, implying that the $i$th quantizer has a single quantization level and is thus inactive. For instance, since the gradients in \eqref{algorithms2} equal zero for $i> K$ and are strictly negative for $i \leq K$, it follows that Algorithm~\ref{alg:greedy} assigns zero bits (i.e., $M_i=1$) for quantizers of index $i > K$. While increasing $P$ to be larger than the spectral support $K$ thus does not affect the \ac{mse} as these quantizers will no be active, using less than $K$ quantized samples may result in different achievable \ac{mse} values. Consequently, in the following corollary we identify the number of quantized samples $P$ which minimizes the \ac{mse} without setting any quantizer to be inactive:
%
\begin{corollary}
\label{co:num}
The \ac{mse} minimizing number of actively quantized samples $P^*$ satisfies  $l_{P^*} \le \log_2M < l_{P^*+1}  $, where
\begin{equation}
l_i = \begin{cases}
{1+\sum_{j=1}^i \log_2 \lceil \tilde l_{j,i} \rceil},& i \le K,\\
{+\infty,}& i > K.
\end{cases}
\label{eqn:num}
\end{equation}
Here, $\tilde l_{j,i} \in \mathbb{R}$ is given by the solution to the following equality:
\begin{equation}
g_j(\tilde l_{j,i}) = g_i(1),
\end{equation}
where $g_i(\bullet)$ is $i$-th element in the gradient of \eqref{AchMSE} w.r.t. $M_i$.
\end{corollary}

\begin{IEEEproof}
The proof is given in Appendix~\ref{app:co2}.
\end{IEEEproof}

\smallskip
Corollary~\ref{co:num} provides an increasing sequence, which divides the feasible interval of $M$ decision levels into $K$ sub-intervals. Before designing the graph signal compression system, we can determine the expected number of quantizers according to the size of $M$. When the total number of bits is large enough satisfying $\log_2M \ge l_K$, $K$ samples are actively quantized;
When the total number of bits is limited, as discussed above, some inactive quantizers can be removed. In such a case, the reduction in the number of quantized samples also reduces the computational complexity of Algorithm~\ref{alg:greedy} since it needs to compute $P$-length gradient vector in each iteration. However, for a specific graph signal compression task, if the number of samples that can be taken $P$ is less than $P^*$ of Corollary~\ref{co:num}, then all samples should be quantized with at least two bits. 



\vspace{-0.2cm}
\section{Joint Sampling and Quantization Using Frequency-Domain Graph Filters}
\label{sec:constrained}
\vspace{-0.1cm}



Section~\ref{sec:unconstrained} designs the compression rule without  imposing any constraints on the sampling matrix, i.e., setting ${\bf \Psi} = {\bf I}_{\mathcal{S}} {\bf F}$, thus allowing the graph filter ${\bf F}$ to be any $N\times N$ matrix.
In this section, we specialize our analysis to 
frequency-domain graph filters. We first present in Subsection~\ref{subsec:constrainedP1} an alternative problem formulation, obtained by restricting the graph filter in \eqref{problem2} to implement frequency-domain graph filtering. 
	Based on the modified problem formulation, we propose an algorithm to tune the bit allocation and sampling set design in Subsection~\ref{subsec:BitSam}, after which we provide an alternating optimization algorithm to obtain the overall compression scheme  in Subsection~\ref{subsec:Overall}, and discuss the resulting design in Subsection~\ref{subsec:Discussion}.

\vspace{-0.2cm}
\subsection{Problem Formulation}	
\label{subsec:constrainedP1}
\vspace{-0.1cm}
To formulate the compression problem using frequency-domain graph filters, we return to  \eqref{problem2}, which serves as a basis for the design of the unconstrained system in Section~\ref{sec:unconstrained}.
The resulting formulation is stated in the following lemma:
\begin{lemma}
\label{lem:Prob2x}
When the sampling matrix ${\bf \Psi}$ is restricted to represent frequency-domain graph filtering, i.e., ${\bf \Psi} = {\bf I}_\mathcal{S} {\bf U} F({\bf \Lambda}) {\bf U}^T $, then, by defining ${\bf H}\triangleq  {\bf{U}}{{F}}({\bf \Lambda}) \tilde{\bf \Lambda}^2  {{F}}({\bf \Lambda}) {\bf{U}}^T$ and ${\bf{X}} \triangleq {\bf{U}}{{F}}({\bf \Lambda}){\tilde{\bf \Lambda}}{{F}}({\bf \Lambda}){{\bf{U}}^T}$, \eqref{problem2} is specialized into
\vspace{-0.1cm}
\begin{equation}
\label{problem2x}
\tag{P3}
\begin{aligned}
\mathop {\max }\limits_{{{\bf{I}}_{{\mathcal S}}},{F}({\bf \Lambda}), \{M_i\}} &{\rm{Tr}}\left( {{\bf{I}}_{{\mathcal S}} {\bf H} {\bf{I}}_{{\mathcal S}}^T } \left({\bf{I}}_{{\mathcal S}} {\bf X} {\bf{I}}_{{\mathcal S}}^T + {\bf G}\right)^{-1} \right), \\
\text{s.t.}~~~ &\sum_{i=1}^P \log_2 M_i \le \log_2M, \, M_i \in \mathbb{Z}^+, i \in \mySet{P}.
\end{aligned}
\vspace{-0.1cm}
\end{equation}
\end{lemma}
\begin{IEEEproof}
The proof is given in Appendix~\ref{app:proof4}.
\end{IEEEproof}



\smallskip
Problem~\eqref{problem2x} specializes \eqref{problem2} to sampling using frequency-domain graph filtering. It translates the compression system design into the joint optimization of the frequency-domain graph filter ${\bf F} = {\bf U}F({\bf \Lambda}){\bf U}^T$, the sampling set selection ${\bf{I}}_{{\mathcal S}}$, and the bit allocation $\{M_i\}$. In particular, $F({\bf \Lambda})$ affects the matrices  ${\bf H}$ and ${\bf X}$, while the bit setting is implicitly encapsulated in the distortion matrix ${\bf G}$ \eqref{cons2}.
Problem \eqref{problem2x} is non-convex due to the coupling its variables. To tackle this challenging optimization, in the following subsections, we first show how one can tune ${\bf{I}}_{{\mathcal S}}$ and $\{M_i\}$ for a given graph filter $\bf F$, based on which we propose an alternating optimization algorithm to jointly set the compression system parameters. 

\vspace{-0.2cm}
\subsection{Sampling Set and Bit Allocation Design}
\label{subsec:BitSam}
\vspace{-0.1cm}
Here, we solve  \eqref{problem2x} under a given filter $\bf F$  to obtain the 
sampling set $\mathcal{S}$ and bit allocation $\{M_i\}$. While the number of samples taken is $|{\mathcal S}| = P$, the number of actively quantized samples can be smaller, as revealed by Corollary~\ref{co:num}.
Consequently, while the number of samples taken is $P$, which is fixed here, the optimization of the sampling set should account for the fact that possibly less than $P$ samples are actively quantized. Note that this requirement was not encountered when considering unconstrained graph filters, where ${\bf \Psi}$ can be any $P\times N$ matrix, while here ${\bf \Psi}$ needs to be explicitly divided into sampling set selection ${\bf I}_{\mySet{S}}$ and the graph filter ${\bf F}$.

To optimize the sampling set in light of the above consideration, we allow $\mySet{S}$ to contain only the actively quantized, i.e., $|\mySet{S}|\leq P$. To formulate the resulting problem, we define
\vspace{-0.1cm}
\begin{equation} \label{virtualM}
{\tilde M}_i \triangleq \begin{cases}
{ M_j  }&{  i = (\mathcal{S})_j ,  }\\
1 & {\text{otherwise}},
\end{cases}
\vspace{-0.1cm}
\end{equation}
where $(\mathcal{S})_j$ is the $j$-th sampling vertex in $\mathcal{S}$.
Note that in \eqref{virtualM} there is a one-to-one correspondence between $\{ M_j \}$, $\mathcal{S}_j$ and $\{{\tilde M}_i\}$, i.e., $\{ M_j \}$ and $\mathcal{S}$ are uniquely determined by $\{{\tilde M}_i\}$. Therefore, by letting ${\bf G}_{\{ \tilde M_j \}}$  be the matrix  $\bf G$ in \eqref{cons2}  with the bit allocation n $\{ \tilde M_j \}$, the joint optimization of the sampling set and the bit allocation for a given ${\bf F}$ can be formulated as
\vspace{-0.1cm}
\begin{align}
\mathop {\max }\limits_{\{\tilde M_i\}}\, &{\rm{Tr}}\left( {{\bf{I}}_{{\mathcal S}} {\bf H} {\bf{I}}_{{\mathcal S}}^T } \left({\bf{I}}_{{\mathcal S}} {\bf X} {\bf{I}}_{{\mathcal S}}^T + {\bf G}_{\{\tilde M_i\}}\right)^{-1} \right),  \notag \\
\text{s.t.}~~~ &\sum_{i=1}^P \log_2 \tilde M_i \le \log_2M,\, \tilde{M}_i \in \mathbb{Z}^+, i\in\mySet{P},  \label{problem5}\\
& \mathcal{S} = \{j|\tilde M_j >1\}, |\mathcal{S}| \le P. \notag
\vspace{-0.1cm}
\end{align}

To solve \eqref{problem5}, we propose a greedy method starting with  ${\tilde M}_i^{(0)} = 1$ for each $i \in \mySet{N} \triangleq \{1,2,\ldots,N\}$ and select a quantizer to assign an additional level, as in Algorithm~\ref{alg:greedy}.
In the $k$-th iteration, we define the sampling set and bit allocation as ${\mathcal S}^{(k)}$ and $\{\tilde M_i^{(k)}\}$  respectively. The objective  \eqref{problem5} is now 
\begin{equation*}
\!f(\{\tilde M_i^{(k)}\})\! \triangleq \! {\rm{Tr}}\left( {{\bf{I}}_{{\mathcal S}^{(k)}} {\bf H} {\bf{I}}_{{\mathcal S}^{(k)}}^T }\! \left({\bf{I}}_{{\mathcal S}^{(k)}} {\bf X} {\bf{I}}_{{\mathcal S}^{(k)}}^T\! +\! {\bf G}_{\{\tilde M_i^{(k)}\}}\right)^{-1} \right).
\end{equation*}
%
In each iteration, we select which $\tilde{M}_i$ to increment by approximating the derivative of $f(\cdot)$ w.r.t. each $\tilde{M}_i$. The derivative w.r.t. the integer $\tilde{M}_i$ is approximated by
\begin{equation*}
\frac{\partial f(\{\tilde M_i^{(k)}\})}{\partial \tilde{M}_j} \approx
q_j^{(k)} \triangleq 
{ f(\{\tilde M_i^{(k)} + \boldsymbol{1}_{i=j}\}) - f(\{\tilde M_i^{(k)}\})}.
\end{equation*}
 At the $k$th iteration, the approximated gradient vector ${{\bf q} ^{(k)}}$ is
\begin{equation}\label{gradientq}
\nabla_{\{\tilde M_i^{(k)}\}}f(\{\tilde M_i^{(k)}\}) \approx {\bf q}^{(k)} \triangleq [  q_1^{(k)},  q_2^{(k)},..., q_N^{(k)} ]^T.
\end{equation}






Our proposed greedy method iteratively updates $\{\tilde{M}_i\}$ while accounting for the constraints on the overall number of bits and the maximal number of actively qauntized samples imposed in \eqref{problem5}.
To that aim we define a set of possible sample indices $\mySet{I}$, which is initialized to $\mySet{N}$, i.e., all possible samples. In each iteartion,  we choose one quantizer index in $\mySet{I}$ to assign an additional level, by selecting the one which we expect to contribute most substantially to the objective in \eqref{problem5},  determined by the largest entry of \eqref{gradientq}. Since the number of actively quantized samples cannot go beyond $P$, once $P$ different samples are actively quantized, we fix the set $\mySet{I}$ to only consider the indexes of these samples subsequently.
The overall process is summarized in Algorithm~\ref{alg:greedy2}.

\begin{algorithm}
\caption{Bit allocation for a given graph filter}
\label{alg:greedy2}
{\bf Input:}
Graph filter matrix $\bf F$, maximal number of samples $P$;\\
{\bf Output:}
Bit allocation $M_i$ and Sampling set $\mathcal{S}$;\\
{\bf Initialize:}
$k=0$, bit allocation $\tilde M_i^{(0)}=1$ for $\forall i \in \mySet{N}$, possible indices $\mySet{I} = \mySet{N}$, and sampling set $\mathcal{S}^{(0)} = \emptyset$.
\begin{algorithmic}[1]
\While {$\sum_i \log_2\tilde M_i^{(k)} \le \log_2M$ }
    \State Compute ${q }_i^{(k)}$ for each $i\in\mySet{I}$;
    \State Update $\tilde M_i^{(k+1)}=\tilde M_i^{(k)}+\boldsymbol{1}_{i=\arg \min_{i\in \mySet{I}} q_i^{(k)}}$;
    \If{$|\mySet{S}^{(k)}| = P$}
    \State Set possible indices $\mySet{I} = \mySet{S}^{(k+1)} = \mySet{S}^{(k)}$;
    \Else
    \State Update $\mathcal{S}^{(k+1)} = \{ i| \tilde M_i^{(k+1)}>1 \}$;
    \EndIf
    \State $k=k+ 1$;
    \EndWhile
\State \Return $\{\tilde M_i\} = \{\tilde M_i^{(k-1)}\} $, $\mathcal{S} = \mathcal{S}^{(k-1)}$.
\end{algorithmic}
\end{algorithm}

\vspace{-0.2cm}
\subsection{Alternate Optimization Algorithm Design}
\label{subsec:Overall}
\vspace{-0.1cm}
We  proceed to solve  \eqref{problem2x} for a given bit allocation $\{M_i\}$ and sampling set $\mathcal{S}$. In this case, \eqref{problem2x} is simplified to 
\begin{equation}\label{problem3x}
\mathop {\max }\limits_{F({\bf \Lambda})}\, {\rm{Tr}}\left( {{\bf{I}}_{{\mathcal S}} {\bf H} {\bf{I}}_{{\mathcal S}}^T } \left({\bf{I}}_{{\mathcal S}} {\bf X} {\bf{I}}_{{\mathcal S}}^T + {\bf G} \right)^{-1} \right).
\end{equation}
The dependency of \eqref{problem3x} on $F({\bf \Lambda})$ is encapsulated in  ${\bf H}$ and ${\bf X}$.
Recalling the definition of the frequency domain graph filter in (Definition \ref{def:filter}) and the imposed polynomial structure, i.e., ${{F}}({\bf \Lambda}) = \sum_{i = 0}^{K_0} \beta_i {\bf \Lambda}^i$, we aim to optimize \eqref{problem3x} with respect to the  coefficients $\{ \beta_i\}$. Since ${{F}}({\bf \Lambda})$ is diagonal, we do this by first optimizing its diagonal entries, denoted $\{\tilde{\lambda}_i\}_{i=1}^N$, after which we approximate them using a set
of $\{\beta_i\}_{i=1}^{K_0}$.

To solve \eqref{problem3x} with respect to 
$\{ \tilde\lambda_i\}$ 
, we consider the alternate optimization method. In particular, we optimize each ${\tilde \lambda _i}$ for fixed $\{\tilde{\lambda}_j\}_{j\neq i}$, repeating and updating this process for each $i$ until convergence is achieved.
Now, for a fixed index $i\in \mySet{N}$, it follows from the definitions of ${\bf X}$ (see Lemma~\ref{lem:Prob2x}) and ${\bf G}$ in \eqref{cons2} that  ${\bf{I}}_{{\mathcal S}} {\bf X} {\bf{I}}_{{\mathcal S}}^T + {\bf G} = { {\tilde\lambda _i^2{\bf{B}} + {\bf{C}}} }$, where the  $|\mySet{S}|\times|\mySet{S}|$  matrices ${\bf B}$ and ${\bf C}$ are independent 
of $\tilde{\lambda}_i$, and are defined as
\begin{align*}
\big({\bf{B}}\big)_{p,q} &\triangleq m_{p,q}\left(\sigma_i^2+\sigma^2\right)({\bf U})_{(\mySet{S})_p,i}({\bf U})_{(\mySet{S})_q,i}, \\
\big({\bf{C}}\big)_{p,q} &\triangleq m_{p,q}\sum\limits_{j = 1,j \ne i}^N \tilde\lambda _j^2{{\left( {\sigma _j^2 + {\sigma ^2}} \right)}}({\bf U})_{(\mySet{S})_p,j}({\bf U})_{(\mySet{S})_q,j},
\end{align*}
with $m_{p,q}\triangleq 1 +\frac{{2\eta^2 }}{{3M_p^2}}{\boldsymbol{1}}_{p=q}$. We can now write (and solve) the optimization of \eqref{problem3x} w.r.t. a single $\tilde{\lambda}_i$ as
\begin{equation}\label{problem5x}
\mathop {\max }\limits_{\tilde \lambda_i} \, {\rm{Tr}}\left( {{\bf{I}}_{{\mathcal S}} {\bf H} {\bf{I}}_{{\mathcal S}}^T } \left( \tilde\lambda _i {\bf{B}} + {\bf{C}}   \right)^{-1} \right).
\end{equation}
\begin{lemma}
\label{lem:Prob5x}
The solution to \eqref{problem5x} is $\tilde \lambda_i^* = \sqrt{\hat \lambda_i^*}$, where
\begin{equation}\label{ComputeLambda}
{\hat \lambda_i^*} = \min_{\hat \lambda_i \geq 0} \sum_{j = 1}^P \frac{ a_j }{\hat \lambda_i + {({\bf{\Lambda }}_x)}_{j,j}}.
\end{equation}
In \eqref{ComputeLambda}, we define $a_i \triangleq ({\bf A})_{i,i}^2  {({\bf{\Lambda }}_x)}_{j,j} - \sum_{n=1,n\neq i}^N \hat \lambda_n ({\bf A})_{i,j}^2$, where ${\bf A} \triangleq {\bf U}_x^T{\bf{B}}^{ - 1/2}{\bf I}_{\mathcal{S}}{\bf U}{\tilde{\bf \Lambda}}$,
with ${\bf U}_x$ and ${\bf{\Lambda }}_x$ being the eigenvectors and eigenvalues of  ${{\bf{B}}^{ - 1/2}} {{\bf{C}}}{{{\bf{B}}}^{ - 1/2}}$, respectively.
\end{lemma}

\begin{proof}
The proof is given in Appendix~\ref{app:proof5}.
\end{proof}


Problem \eqref{ComputeLambda} is a concave-convex fractional programming problem, which can be efficiently
solved using the quadratic transform\cite{quadratic2018}. After $\tilde\lambda_i^*$ is obtained, we update the value of ${ F}({\bf \Lambda})$ and optimize the remaining $\tilde\lambda_i$.  This procedure is repeated until a pre-specified convergence condition $\epsilon$ is met, or a maximal number of iterations $T$ is reached. The resulting alternating algorithm for constrained graph filter design is summarized in Algorithm ~\ref{alg:alternative}, where $\epsilon$ is the threshold and $T$ is the maximum number of 
iterations.

\begin{algorithm}
\caption{Constrained graph filter for compression design}
\label{alg:alternative}
{\bf Input:}
Sampling set $\mathcal{S}$ and bit assignment $\{ M_i \}$;\\
{\bf Output:}
Graph filter matrix $\bf F$;\\
{\bf Initialize:}
$\tilde\lambda_i^{(0)} = 1, i \in  \mySet{N}$.
\begin{algorithmic}[1]
\For {$k=1,2,\ldots, T$}
\For {$i \in \mySet{S}$}
\State Update $\tilde\lambda_i^{(k)} = \tilde\lambda_i^*$ computed via Lemma~\ref{lem:Prob5x};
\EndFor
\If {$\big\|[\tilde\lambda_1^{(k)},\ldots,\tilde\lambda_N^{(k)}]-[\tilde\lambda_1^{(k-1)},\ldots,\tilde\lambda_N^{(k-1)}]\big\|_2^2 \le \epsilon$}
\State Break;
\EndIf
\EndFor
\State Compute $F (\Lambda) = \text{diag}\{\tilde\lambda_1^{(k-1)},\tilde\lambda_2^{(k-1)},...,\tilde\lambda_N^{(k-1)}\}$;
\State \Return ${\bf F} = {\bf U} F ({\bf \Lambda}){\bf U}^T $.
\end{algorithmic}
\end{algorithm}

Finally, the overall alternating optimization based algorithm for setting the frequency domain graph filter, sampling set, and bit allocation based on \eqref{problem2x} is summarized as Algorithm~\ref{alg:overall}. The algorithm alternates between setting the frequency-domain graph filter for fixed sampling set and bit allocation  via  Algorithm~\ref{alg:alternative}, and optimizing the sampling set and bit allocation for fixed graph filter via Algorithm~\ref{alg:greedy2}.

\begin{algorithm}
\caption{Compression system design with frequency domain graph filters}
\label{alg:overall}
{\bf Input:}
Graph Fourier basis matrix ${\bf U}^T$, maximal samples $P$;\\
 {\bf Output:}
Graph filter $\bf F$, bit allocation $\{M_i\}$, sampling set $\mathcal{S}$;\\
 {\bf Initialize:}
${\bf F}^{(0)} = {\bf I}$, $k = 0$.
\begin{algorithmic}[1]

\Repeat
\State Compute $\mathcal{S}^{(k)}$ and $\{M_i^{(k)}\}$ for ${\bf F}^{(k)}$ via Algorithm~\ref{alg:greedy2};
\State Update ${\bf F}^{(k+1)}$ for $\mathcal{S}^{(k)}$ and $\{M_i\}^{(k)}$ via Algorithm~\ref{alg:alternative};
\State $k = k+1$;
\Until
$\|{\bf F}^{(k)} - {\bf F}^{(k-1)}\| \le \epsilon$
\State \Return ${\bf F} =  {\bf F}^{(k)}$, $\{M_i\} = \{M_i^{(k-1)}\}$, $\mathcal{S} = \mathcal{S}^{(k-1)}$  .


\end{algorithmic}
\end{algorithm}

\vspace{-0.4cm}
\subsection{Discussion}
\label{subsec:Discussion}
\vspace{-0.1cm}

Our proposed compression scheme is based on jointly designing the sampling and quantization mappings to achieve a finite bit representation of the graph signal in a manner which allows its accurate reconstruction. This is achieved utilizing a sampling matrix which, instead of selecting a subset of the nodes as in \cite{SamplingTheory}, samples linear combinations of the graph signal elements using frequency domain graph filtering. The increased flexibility in the design of the sampling mapping is exploited to generate a sampling set such that the graph signal spectrum can be accurately recovered from the quantized samples. In particular, we do this by accounting for the statistical moments of the noisy graph signal via the relaxed formulation \eqref{problem2x}  when jointly tuning the sampling matrix and the quantization rule.  We focus here on implementing such compression mechanisms using either unconstrained graph filters as well as structured frequency domain graph filters. However, one can consider additional forms of graph filtering which may be preferable in some applications, e.g., vertex domain graph filters \cite{tanaka2020sampling}. Nonetheless, we leave the specialization of the proposed joint sampling and quantization mechanism to alternative forms of graph filters for future work.




Finally, we note that our design considers a limit on the overall number of bits used by the  quantizers to generate the finite bit representation. A less distorted finite bit representation can be obtained by replacing the uniform scalar quantizers with vector quantization \cite[Ch. 23]{polyanskiy2014lecture}, though such mappings tend to be highly complex, as opposed to simplified uniform quantization. An alternative approach to improve upon the existing mechanism is to  further compress $Q({\bf \Psi}{\bf x})$ via lossless source coding, e.g., entropy coding \cite[Ch. 13]{cover2012elements}, to $Q({\bf \Psi}{\bf x})$, as done in \cite{shlezinger2020uveqfed}. Doing so results in an equivalent representation whose number of bits is dictated by the entropy of $Q({\bf \Psi}{\bf x})$ (which is not larger than $\log_2 M$). This implies that the resulting finite bit representation can be further compressed, and motivates extending our analysis to constrain the entropy of $Q({\bf \Psi}{\bf x})$ rather than its bit budget, i.e., via entropy-constrained quantization \cite{gray1998quantization}, which we leave for future research.

%

	\vspace{-0.2cm}
	\section{Numerical Evaluations}
	\label{sec:Sims}
	\vspace{-0.1cm}
	
	In this section, we evaluate the performance of the proposed joint sampling and quantization methods for graph signal compression.\footnote{The source code used in this section is available online in the following link: \url{https://github.com/Pei65536/Graph_Signal_Compression.git}
	}  In Subsection~\ref{subsec:artificial}, we simulate graph signals representing measurements taken using parameters provided in GSP toolbox\cite{GSPBOX}. We then use the proposed scheme to compress graph signals representing real-world temperature data in Subsection~\ref{subsec:real}.
Finally, in Subsection~\ref{subsec:image}, we apply the proposed sampling algorithm to image compression.

\begin{figure}
\centering
\includegraphics[width=\figWidth,height=\figheight]{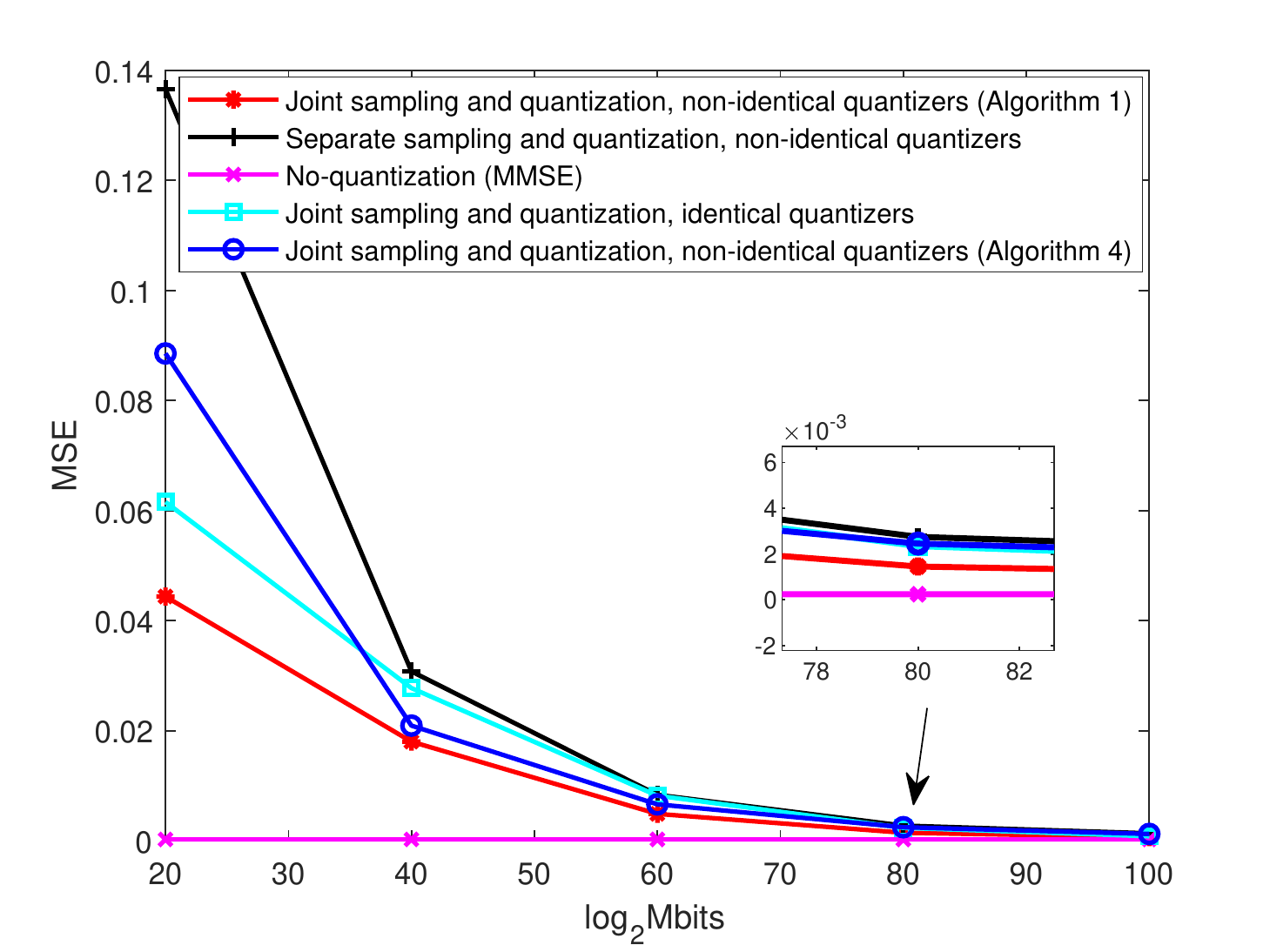}
\vspace{-0.4cm}
\caption{MSE versus the number of bits for different schemes. 
}
\label{fig:MSEvsBits}
\end{figure}

Throughput this section, we evaluate the MSE achieved when compressing the graph signal using  to the following schemes:  $(1)$ Joint sampling and quantization with unconstrained filters via Algorithm~\ref{alg:greedy}; $(2)$ Separate sampling and quantization with non-identical quantizers as in \cite{QuantizationSampling};  $(3)$ the \ac{mmse} estimate achievable of ${\bf U}_k {\bf c}$ which is achievable with infinite resolution quantization, i.e., without quantization constraints; $(4)$ Joint sampling and quantization with identical quantizers as in \cite{HardwareLimited}; and $(5)$ Joint sampling and quantization with frequency domain graph filters via Algorithm~\ref{alg:overall}.

\vspace{-0.2cm}
\subsection{Synthetic Data}	
\label{subsec:artificial}
\vspace{-0.1cm}
We begin with synthetic simulated scenario. Here, the graph signals represent measurements acquired by a sensor network comprised of $N=100$ nodes, where each node is connected with its neighbors located within its transmission range, and the bandwidth is  $K=20$. We then take $1000$ different noisy bandlimited graph signals with zero mean and covariance  ${\bf C}_{\bf x}$, in which the non-zero GFT coefficients are randomized from $\mathcal{N} (0, 1/\lambda_i)$, where $\lambda_i$ is the $i$th eigenvalue of $\bf L$.

We first numerically evaluate the MSE in compressing the graph signals versus the bit budget $\log_2M$, which takes values in  $[20, 120]$, while setting the Gaussian noise power to  $\sigma_0^2=-30~\text{dB}$. The results, depicted in Fig. \ref{fig:MSEvsBits}. Observing Fig.~\ref{fig:MSEvsBits},  we note that   Algorithm~\ref{alg:greedy} achieves better performance than using separately designed sampling and quantization as well as applying the joint design with identical quantizers of \cite{HardwareLimited}. The separate sampling and quantization mechanism of \cite{QuantizationSampling}, in which the sampling operation is carried out by mere node selection without accounting for the underlying statistics of the signal, achieves the highest \ac{mse} values here, while the proposed schemes is within a small \ac{mse} gap of less than $0.02$ from the \ac{mmse} for bit budgets as small as $\log_2 M = 40$. We also noted that  Algorithm~\ref{alg:overall} allows achieving \ac{mse} which is comparable to that achieved using unconstrained graph filters for bit budgets satisfying $\log_2 M \ge 40$.

Next, we numerically evaluate the effect of the additive noise power on the compression accuracy. To that aim, we compare the achievable MSEs of the schemes simulated in Fig.~\ref{fig:MSEvsBits}, while fixing the bit budget to be $\log_2M=60$, and letting the noise variance  $\sigma_0^2$ grow from $-30$ dB to $-20$ dB. The results, depicted in Fig.~\ref{fig:MSEvsSNR}, demonstrate that our proposed schemes maintains their improved MSE over the competing schemes for different levels of the noise variance. Our scheme is shown to be notably more robust to the presence of noise compared to the all the reference techniques, except for the MMSE estimate, which is not constrained to any bit budget.

\begin{figure}
\centering
\includegraphics[width=\figWidth,height=\figheight]{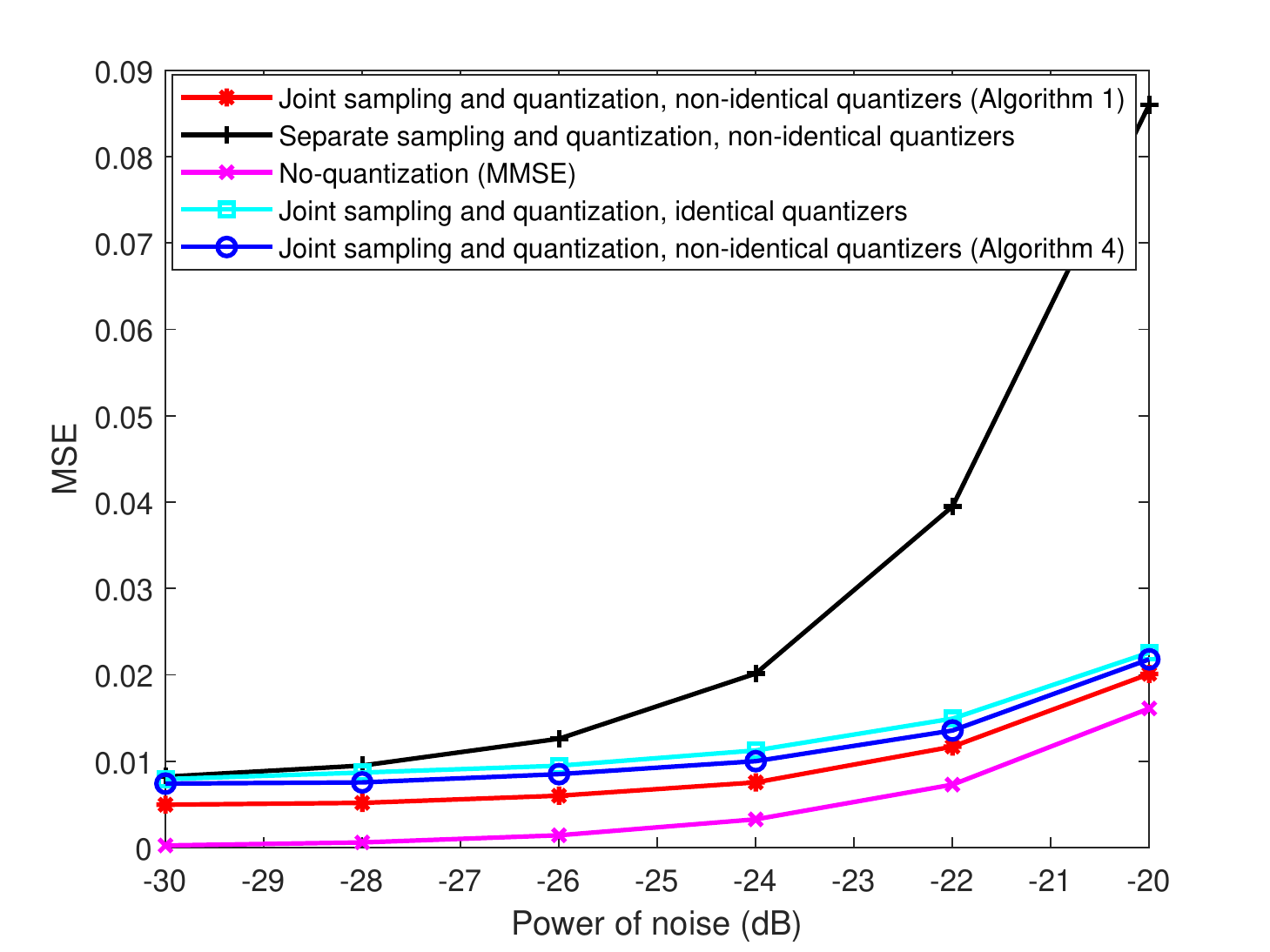}
\vspace{-0.4cm}
\caption{MSE versus the variance of noise.}
\label{fig:MSEvsSNR}
\end{figure}


%

\begin{figure}[htbp]
\centering
\subfigure[]{
\includegraphics[width=1.6in]{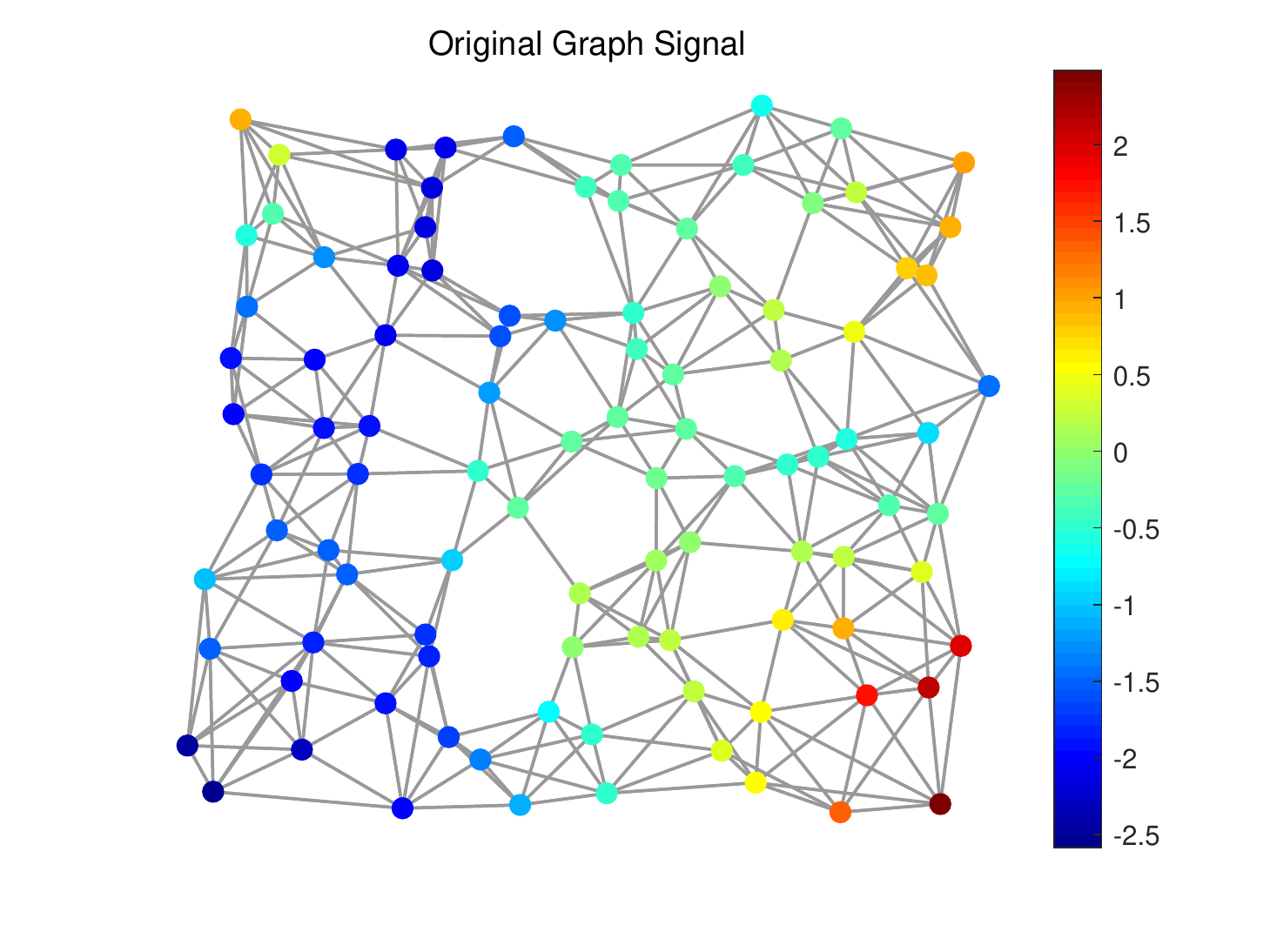}
}
\subfigure[]{
\includegraphics[width=1.6in]{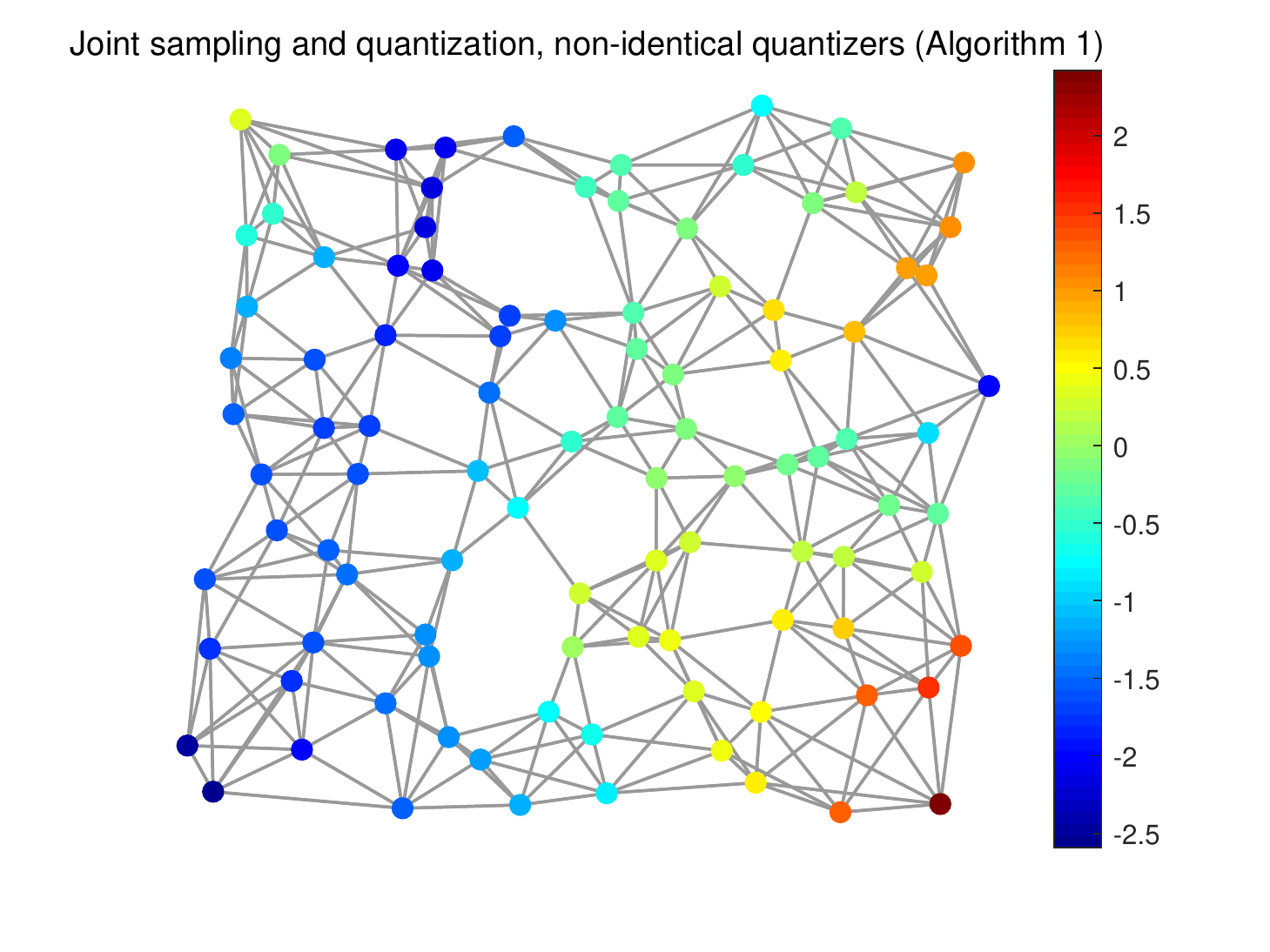}
}
\quad 
\subfigure[]{
\includegraphics[width=1.6in]{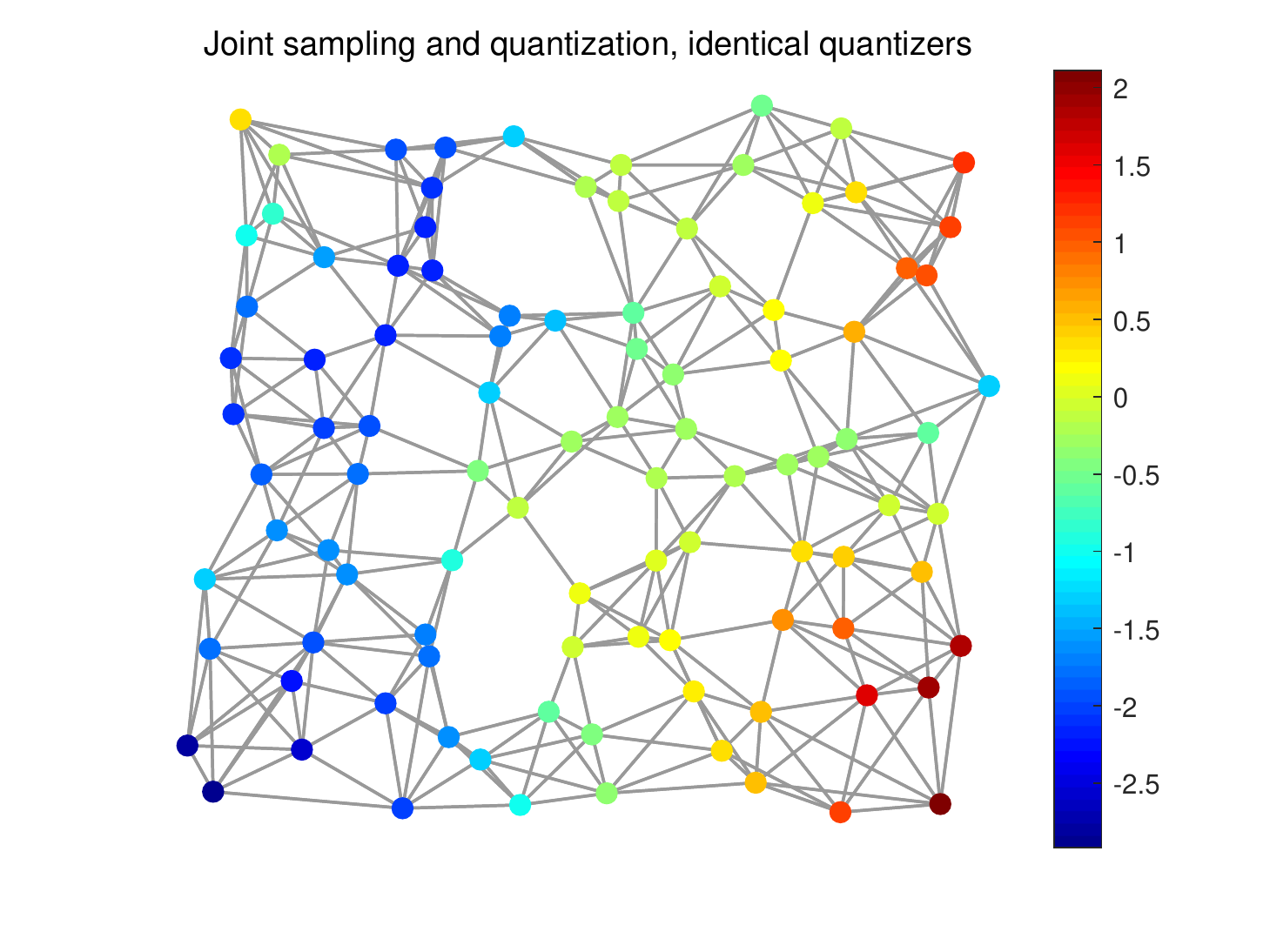}
}
\subfigure[]{
\includegraphics[width=1.6in]{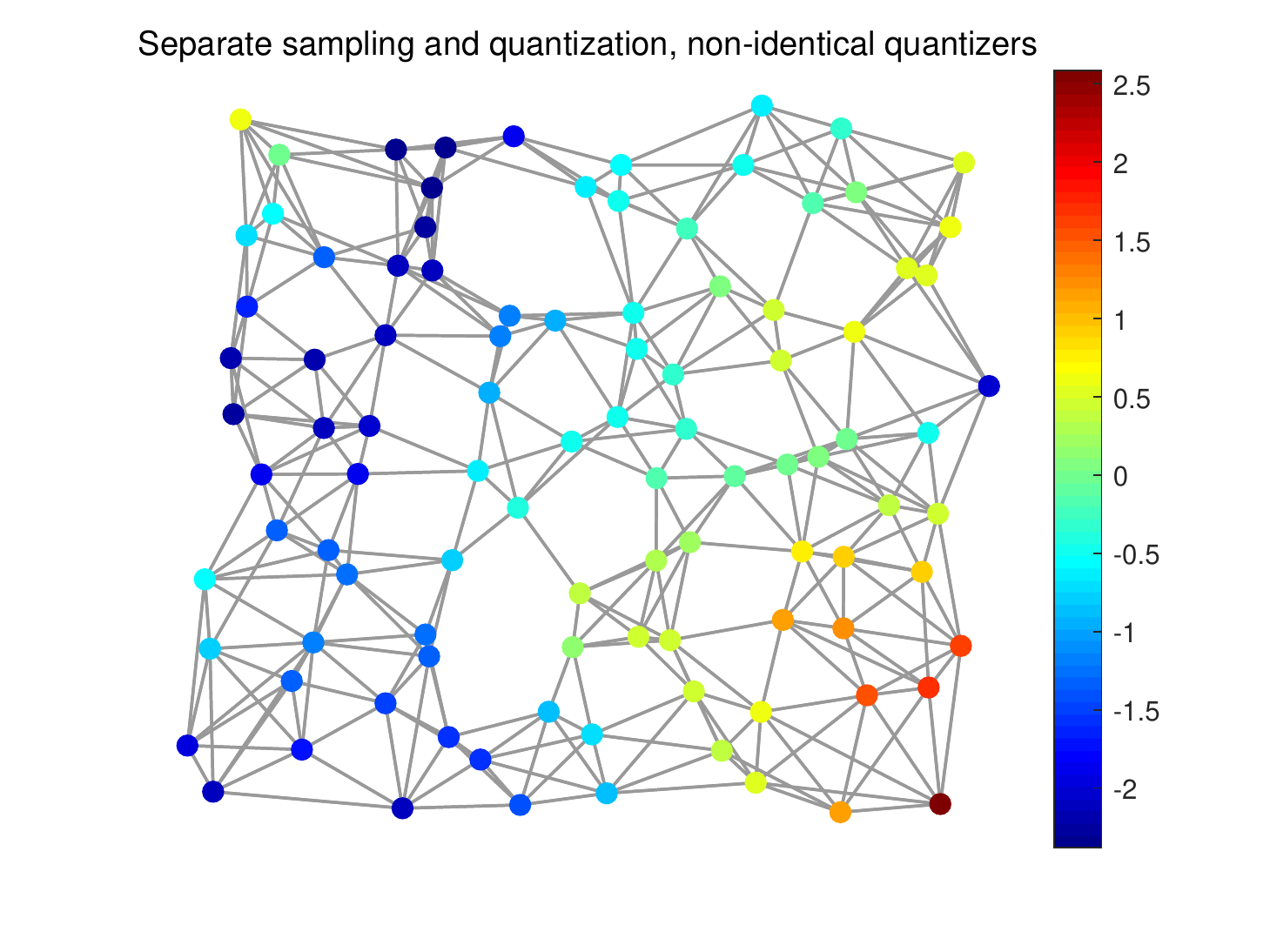}
}
\caption{Synthetic graph signal recovery experiments: $(a)$ Graph signal; $(b)$ Reconstructed signal via Algorithm~\ref{alg:greedy}; $(c)$ Reconstructed signal with separate sampling and quantization with non-identical quantizers; $(d)$ Reconstructed signal with joint sampling and quantization with identical quantizers.}
\label{fig:example}
\end{figure}

To visualize how the observed \ac{mse} gains are translated into a more faithful recovery of compressed graph signals, we depict a realization of a graph signal along with its reconstruction using scheme $(1)$ and the benchmarks $(2)$ and $(4)$ in Fig.~\ref{fig:example}. Here, the noise level is set to $\sigma_0^2 = -30$dB, while the overall bit budget is $\log_2 M = 60$.
Fig.~\ref{fig:example} demonstrates that the proposed scheme based on joint sampling and quantization allows to compress the graph signal into a compact bit representation while resulting in recovered graph filters within a close match of the original graph signal. When using identical quantizers as in \cite{HardwareLimited}, the quality of the recovery signal in Fig.~\ref{fig:example}(d) is degraded and not as smooth as it is in Fig,~\ref{fig:example}(b). This follows since the low frequency components usually occupy a relatively large range, and using identical quantizers limits the recovery of the low frequency components. For the separate sampling and quantization scheme of \cite{QuantizationSampling}, we observe in Fig.~\ref{fig:example}(c) poor recovery on some vertexes, which is caused by the amplification of quantization noise and additive Gaussian noise in restoring unsampled vertexes.

\vspace{-0.2cm}
\subsection{Non-Synthetic Graph Signals}
\label{subsec:real}
\vspace{-0.1cm}
In order to further illustrate the practicability of our  graph signal compression algorithms, we apply the proposed Algorithm~\ref{alg:greedy} and Algorithm~\ref{alg:overall}  to compress non-synthetic graph signals representing temperature data\footnote{The dateset could be downloaded in the following link: \url{http://data.cma.cn}}
The data is obtained from the China Meteorological Data Center, representing measurements taken in January of each year from 2018 to 2021 via  China's meteorological sensors whose distribution is is shown in Fig.~\ref{figstation} (a). The location of each sensor is determined by its latitude and longitude.
It is noted that temperature data distribution depends  not only on the geographical location, but is also affected by the surrounding environment and human actions. Therefore, we construct the graph structure for these measurements based on the smoothness of the data, rather than directly constructing the graph structure based on the distance threshold, e.g., as in $K$-nearest neighbours construction. An example of a resulting graph signal is illustrated in Fig.~\ref{figstation} (b). For the adopted data, the temperature fluctuation tends to a stable state. The statistical information of the temperature data shows that when the bandwidth is greater than 10, the frequency domain component is close to 0. Therefore, we assume that the bandwidth of the data is $K =  10$.

\begin{figure}
\centering
\subfigure[]{
\includegraphics[width=1.6in]{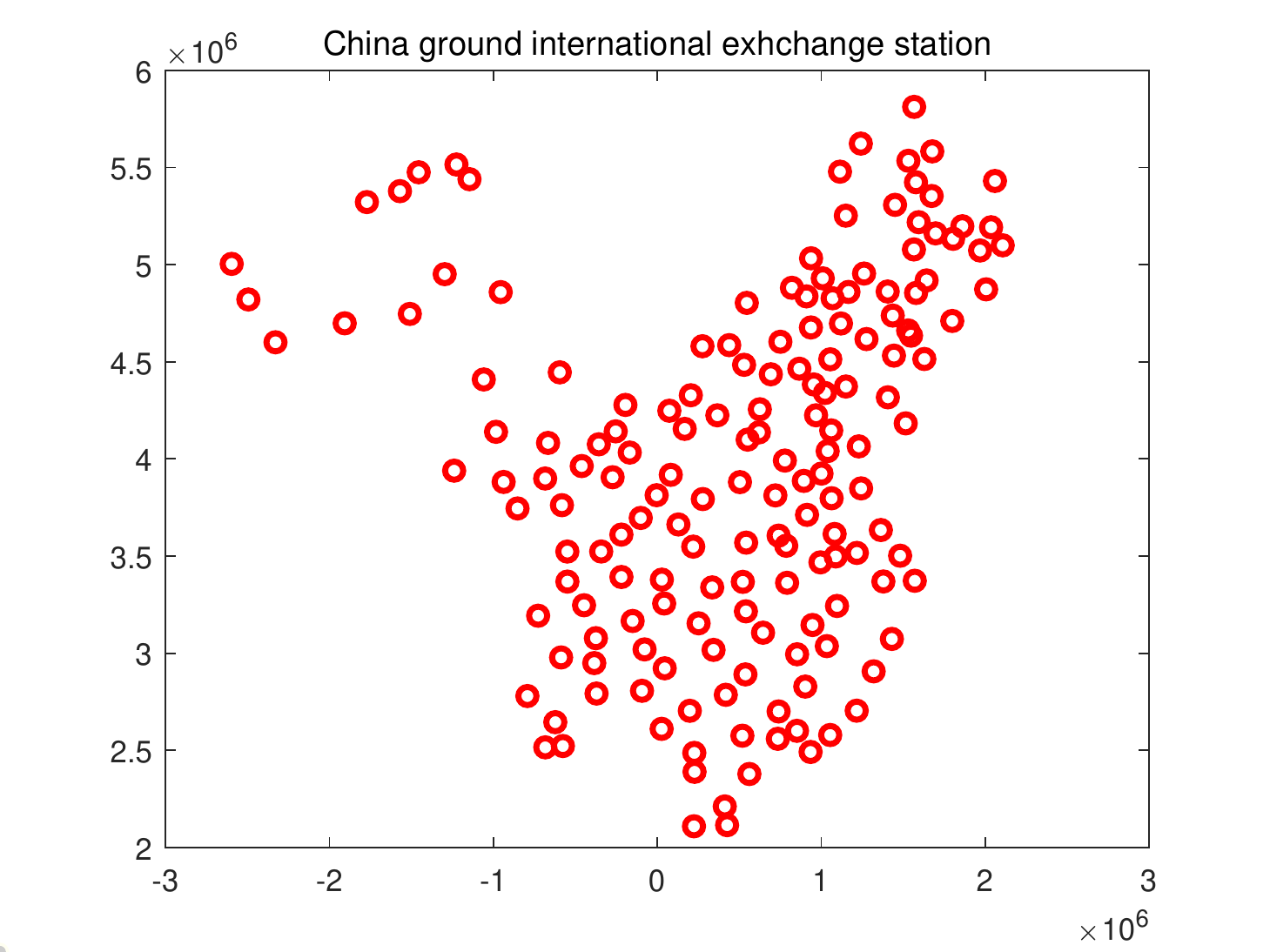}
}
\subfigure[]{
\includegraphics[width=1.6in]{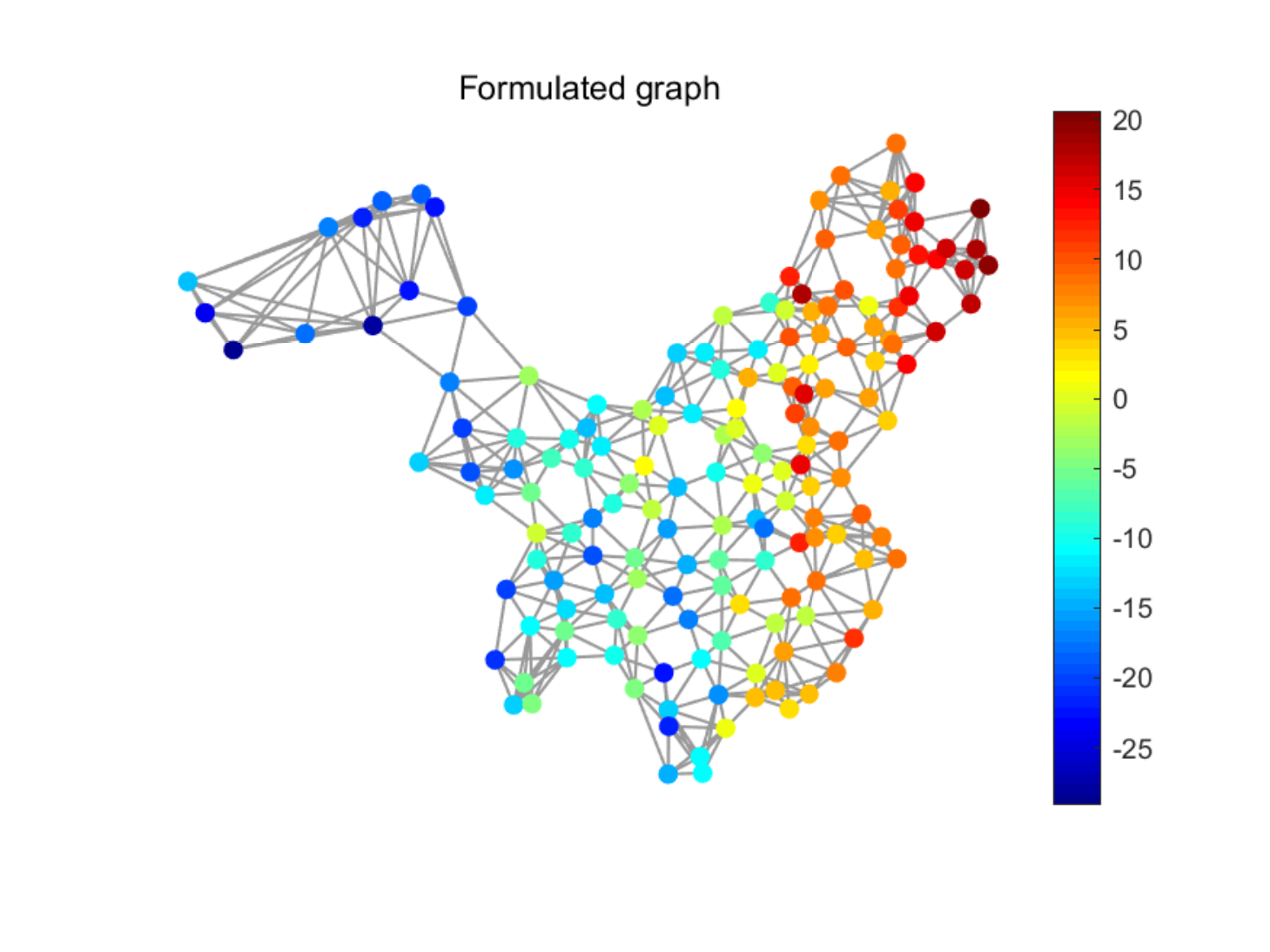}
}
\caption{(a) Geographical location of China ground sensors. (b) A temperature example in formulated graph structure.}
\label{figstation}
\end{figure}

The achievable reconstruction \ac{mse} for these non-synthetic graph signals versus the overall number of bits is depicted in Fig.~\ref{RealDate}. It is noted that in this scenario there is no noise signal that is manually introduced, and the equivalent noise is that which stems from the error in approximating the graph signal as have bandwidth of $K=10$. 
We observe in Fig.~\ref{RealDate} that the \ac{mse} trends versus the total number of bits is preserved as in the synthetic case reported in the previous subsection. Specifically, it is found that compression with unconstrained graph filters via Algorithm~\ref{alg:greedy} outperforms all considered bit-constrained benchmarks, and that similar performance is achieved with frequency domain graph filters for at least $20$ bits using Algorithm~\ref{alg:overall}.

\vspace{-0.2cm}
\subsection{Application in Image Compression}
\label{subsec:image}
\vspace{-0.1cm}
GSP is not only applied to the data with irregular structures, but can also be used to represent l signals such as images, video signals, etc. Here, we apply our proposed Algorithm~\ref{alg:greedy} to image compression. Due to the self-similarity in images, the same or similar structures are likely to recur throughout.
For simplicity, we apply regular line and grid graph topologies, but with unequal edge weights that can adapt to the specific characteristic of an image or a set of images \cite{refimage1}. 

In particular, we use the classic `Lena` image comprised $512 \times 512$ pixels, which are represented as a concatenation from $16384$ graph signals in $4\times 4$ patches. For each graph signal, we approximate its bandwidth to be $K =4$, i.e., we discard high-frequency components, which is relatively quite small. The statistical moments used in our derivation are obtained by averaging over the patches of the image. 
We compare the performance of three image compression schemes: 1) The Discrete Cosine Transform (DCT); 2) Separate sampling with quantization with non-identical quantizers \cite{QuantizationSampling}; 3) The proposed Algorithm~\ref{alg:greedy}. For consistency, we set the block size of DCT to be 
 $4 \times 4$. 
 We set $\log_2M = 64$ 
for each GFT block and  DCT block.
The results, depicted in Fig.~\ref{fig:image}, indicate that the proposed mechanism for joint sampling and quantization which considers graph signals can also be beneficial for image compression, resulting in a faithful recovery of natural images compressed into a low bit representation.


\begin{figure}
\centering
\includegraphics[width=\figWidth,height=\figheight]{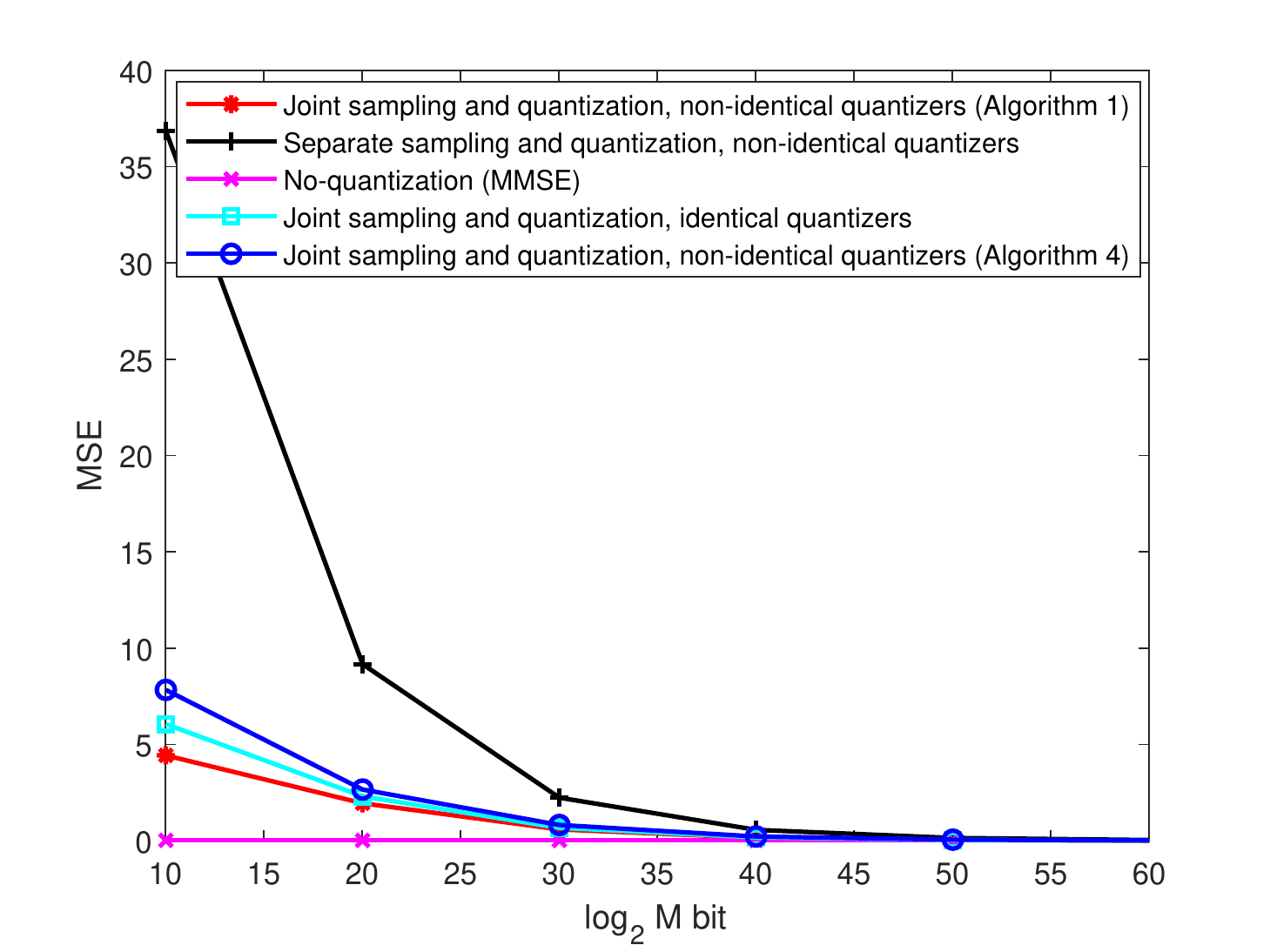}
\vspace{-0.3cm}
\caption{MSE versus the number of bits. 
}
\label{RealDate}
\end{figure}

	\vspace{-0.2cm}
	\section{Conclusions}
	\label{sec:Conclusions}
	\vspace{-0.1cm}
In this work, we studied the compression of bandlimited graph signals into a finite-length sequence of bits by joint sampling and quantization.  
Our derivation is based on the identified similarity between graph signal compression and task-based quantization, which is a framework for configuring \acp{adc} with analog pre-processing.  We formulated the  design problem, which is in general shown to be non-convex. 
We then presented a relaxed formulation considering linear recovery with non-overloaded quantizers. The relaxed problem was used  to design sampling mechanisms for a given allocation of the available bit budget among the quantizers, which in turn is used to derive a greedy bit allocation algorithm. Next, 
we proposed a sampling operator as a row-selection matrix with a graph filter, restricted to combining elements corresponding to the neighbouring nodes of the graph.
We  applied the proposed algorithm to both synthetic and non-synthetic data.  Numerical results show that the proposed scheme compresses high dimensional graph signals into a limited amount of bits while allowing their recovery with an error within a small gap from the MMSE,  achievable with infinite resolution quantziation.

\vspace{-0.2cm}
  \begin{appendix}

  \numberwithin{proposition}{subsection}
\numberwithin{lemma}{subsection}
\numberwithin{corollary}{subsection}
\numberwithin{remark}{subsection}
\numberwithin{equation}{subsection}

\vspace{-0.2cm}
  \subsection{Proof of Proposition \ref{lem:MSE}}
 \label{app:proof1}
\vspace{-0.1cm}
  We first note that the MSE achievable using the linear recovery matrix ${\bf \Phi}^*$ in Lemma \ref{lem:Digital} is given by
  \vspace{-0.1cm}
\begin{align}
\notag
&{\text{MSE}}({\bf \Psi})={\text {Tr}}\left({\bf \Gamma}^*{\bf C}_{\bf x}{\bf \Gamma}^{*T} \right) \\
&\quad -{\text {Tr}}\left({\bf \Gamma}^* {\bf C}_{\bf x} {\bf \Psi}^T \left( {\bf \Psi} {\bf C}_{\bf x} {\bf \Psi}^T + {\bf G}  \right)^{-1} {\bf \Psi} {\bf C}_{\bf x} {\bf \Gamma} ^{*T}  \right).\label{lem12}
  \vspace{-0.1cm}
\end{align}
  By discarding the constant term ${\text {Tr}}\left({\bf \Gamma}^*{\bf C}_{\bf x}{\bf \Gamma}^{*T} \right)$ in \eqref{lem12}, the  matrix ${\bf \Psi}^*$ that minimizes the MSE can be  obtained via
  \vspace{-0.1cm}
\begin{equation}\label{lemmap31}
\begin{aligned}
{{\bf{\hat \Psi }}^*} = \arg {\max _{ \hat \bf{\Psi }}}{\rm{Tr}}\left( {\hat {\bf{\Gamma }}{{\hat {\bf{\Psi }}}^T}{{\left( {\hat {\bf{\Psi }}{{\hat {\bf{\Psi }}}^T} + {\bf{I}}} \right)}^{ - 1}}\hat {\bf{\Psi }}{{\hat {\bf{\Gamma }}}^T}} \right),
\end{aligned}
  \vspace{-0.1cm}
\end{equation}
where $\hat {\bf{\Psi }}\triangleq = {{\bf{G}}^{ - 1/2}}{\bf{\Psi C}}_{\bf{x}}^{1/2},\hat {\bf{\Gamma }} \triangleq = {{\bf{\Gamma }}^*}{\bf{C}}_{\bf{x}}^{1/2}$. Now, by writing the singular value decomposition   $\hat {\bf{\Psi }} = {{\bf{U}}_\Psi }{\bf \Xi}_\Psi {\bf{ V}}_\Psi ^T$, and recalling the definition of $\bf G $, we obtain that
  \vspace{-0.1cm}
$\frac{{\bf G}_{i,i}3M_i^2}{2\eta^2} = \left({\bf \Psi} {\bf C}_{\bf x} {\bf \Psi}^T\right)_{i,i} =  \big({\bf G}^{1/2}\hat {\bf \Psi} \hat {\bf \Psi}^T {\bf G}^{1/2}\big)_{i,i}$,
which results in
\begin{equation}\label{Gcon2}
\begin{aligned}
({{\bf{U}}_\Psi } {\bf \Xi}_\Psi^2{{\bf{U}}_\Psi }^T)_{i,i} = \frac{3M_i^2}{2\eta^2}.
\end{aligned}
  \vspace{-0.1cm}
\end{equation}

\begin{figure}
\centering
\subfigure[]{
\includegraphics[width=1.6in]{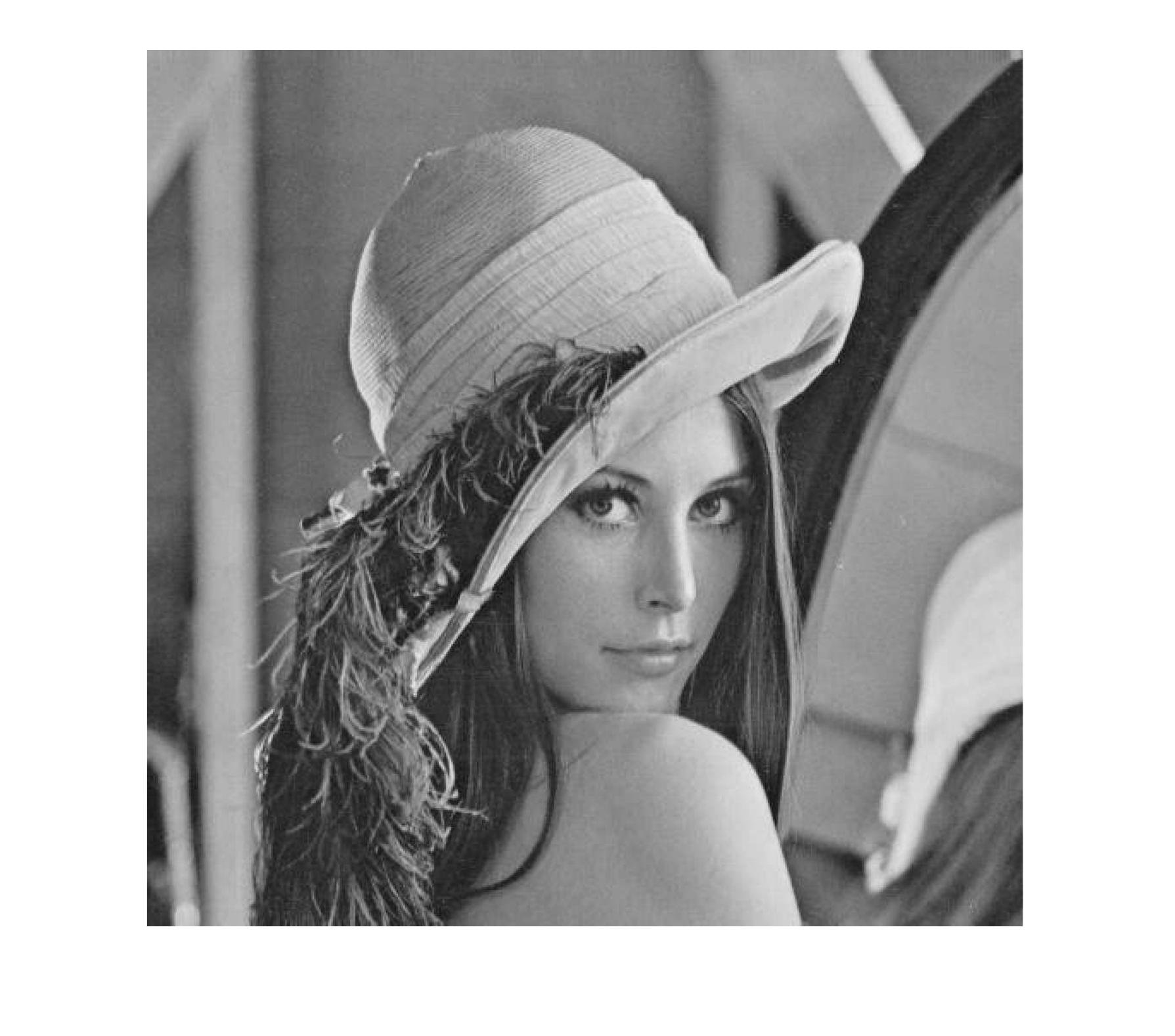}
}
\subfigure[]{
\includegraphics[width=1.6in]{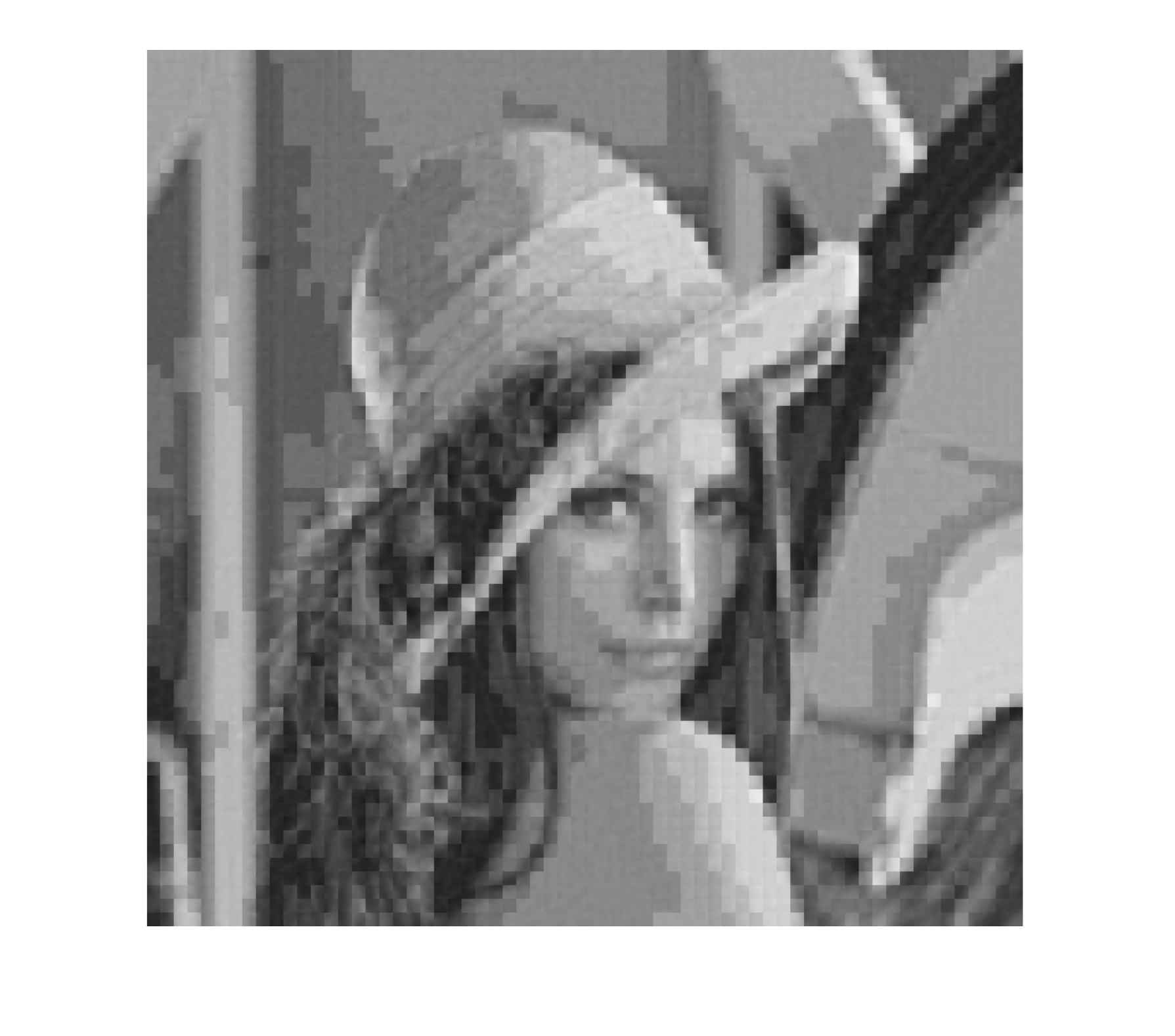}
}
\quad %
\subfigure[]{
\includegraphics[width=1.6in]{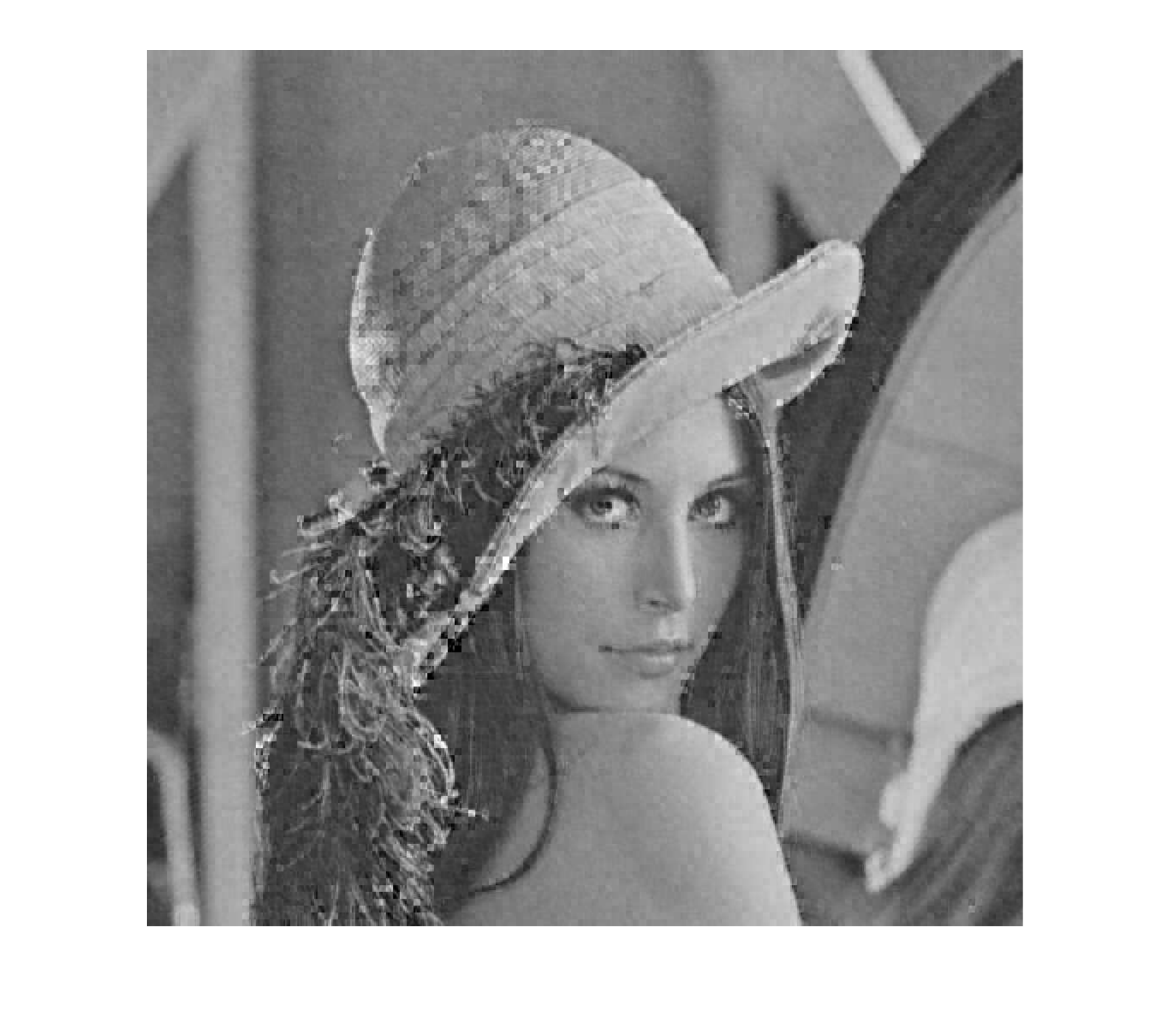}
}
\subfigure[]{
\includegraphics[width=1.6in]{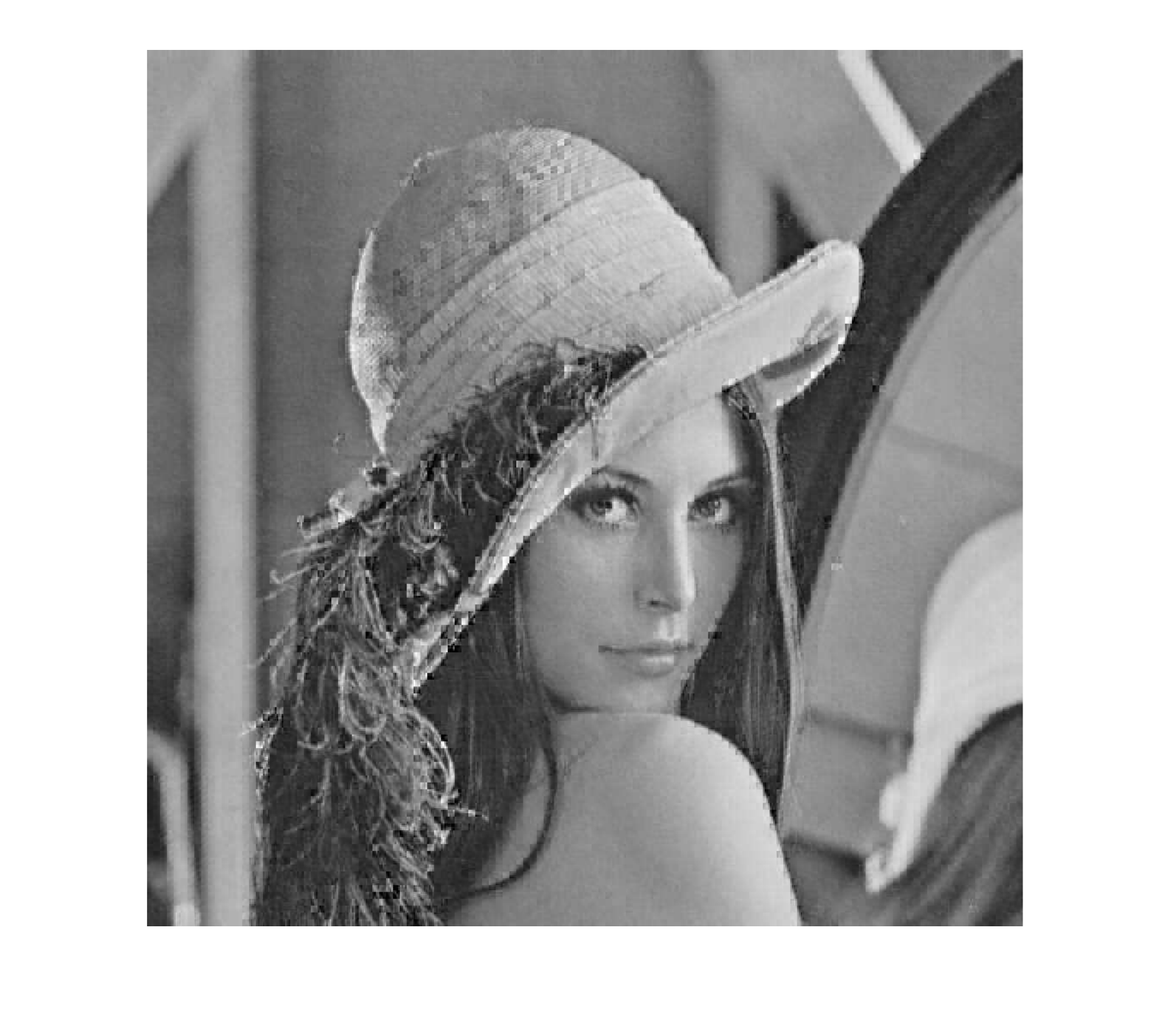}
}
\caption{Image compression. $(a)$ The original image; $(b)$ The compressed image via DCT 
$(c)$ The compressed image via joint sampling and quantization with identical quantizers; $(d)$  The compressed image via Algorithm~\ref{alg:greedy2}.}
\label{fig:image}
\end{figure}

According to majorization theory \cite{majorization}, \eqref{Gcon2} is satisfied if and only if $  \{{3M_i^2}/{2\eta^2}\} \prec   \{\alpha_i\}  $, where $\alpha_i = ({\bf \Xi}_\Psi)_{i,i}^2$. Then problem \eqref{lemmap31} is transformed into
\begin{equation}\label{OptProblem1}
\begin{aligned}
{\max _{{{\bf{V}}_{\Psi}},{\bf{\Xi }}_{\Psi}}} \ &{\rm{Tr}}\left( {{\bf{V}}_\Psi ^T{{\hat {\bf{\Gamma }}}^T}\hat {\bf{\Gamma }}{{\bf{V}}_\Psi }{{\bf{\Xi }}_\Psi^T}{{\left( {{\bf{\Xi }}_\Psi{{\bf{\Xi }}_\Psi^T} + {\bf{I}}} \right)}^{ - 1}}{\bf{\Xi }}_\Psi} \right),\\
{\rm s.t.} \ &{\rm{         }}\left\{ {\frac{{3M_i^2}}{2\eta^2}} \right\} \prec \left\{ \alpha_i \right\}, \alpha_i = ({\bf \Xi}_\Psi)_{i,i}^2.
\end{aligned}
\end{equation}

We note that ${{\bf{\Xi }}_\Psi^T}{{\left( {{\bf{\Xi }}_\Psi{{\bf{\Xi }}_\Psi^T} + {\bf{I}}} \right)}^{ - 1}}{\bf{\Xi }}_\Psi$ is a diagonal matrix with diagonal elements $\frac{\alpha_i}{\alpha_i + 1 }, i \in \mySet{P}$, which are a descending sequence since $\{\alpha_i\}$ is arranged in descending order. Then, the optimal ${\bf{V}}_\Psi$ should be the right singular vectors matrix of  ${\hat {\bf{\Gamma }}}$. By substituting \eqref{ex2} and \eqref{ex2p1} into the definition of ${\hat {\bf{\Gamma }}}$, we obtain ${\hat {\bf{\Gamma }}} = \hat{\bf \Lambda}  (\tilde {\bf \Lambda}^{-1})_K  {\bf U}^T$. Note that $\hat{\bf \Lambda}  (\tilde {\bf \Lambda}^{-1/2})_K$ is a diagonal matrix with descending diagonal elements,
i.e., ${\bf{V}}_\Psi ^T = {{\bf{U}}^T}$ .
Consequently, $ \tilde {\bf{\Gamma }} = {\bf{V}}_\Psi ^T{{\hat {\bf{\Gamma }}}^T}\hat {\bf{\Gamma }}{{\bf{V}}_\Psi }$ is a diagonal matrix, whose diagonal elements arrange in a descending order.
The remaining optimization of \eqref{OptProblem1} is thus
\vspace{-0.1cm}
\begin{align}\label{OptProblem2}
\mathop {\max }\limits_{{\alpha _i}} &\sum_{i=1}^P\frac{{{\lambda _{{\bf{\Gamma }},i}^2}{\alpha _i}}}{{{\alpha _i} + 1}},\\
{\rm s.t}.&{\rm{    }}\sum\limits_{i = 1}^P {{\alpha _i}}  = \sum\limits_{i = 1}^P {\frac{{3M_i^2}}{2\eta^2}}, \,
\sum_{i=1}^p  { {{\alpha _i}} }  \ge  \sum_{i=1}^p{ {\frac{{3M_i^2}}{2\eta^2}} },  \forall p \in \mySet{P}, \notag
\vspace{-0.1cm}
\end{align}
where $\lambda _{{\bf{\Gamma }},i}$ is the $i$-th singular value of ${\hat {\bf{\Gamma }}}$. The constraints in \eqref{OptProblem2} are the expanded form of $\left\{ {\frac{{3M_i^2}}{2\eta^2}} \right\} \prec^w \left\{ \alpha_i \right\}$. We note that \eqref{OptProblem2} is convex.
Once $\{\alpha_i\}$ are obtained, the resulting sampling matrix is ${\bf{\Psi }}^* = {{\bf{G}}^{ 1/2}} {{\bf{U}}_\Psi }{\bf \Xi}_\Psi {\bf{ V}}_\Psi ^T  {\bf{ C}}_{\bf{x}}^{-1/2}$. While the entries of ${\bf G}$ depend on ${\bf{\Psi }}$, we can still obtain the sampling matrix by letting $g_i = {\bf G}_{i,i}^{1/2}$, and substituting the expression of ${\bf{\Psi }}^*$ into the definition of $\bf G$. This results in
$\frac{2 \eta^2 \left({{\bf{ \Psi }}} {\bf{ C}}_{\bf{x}} {{\bf{ \Psi }}^T}\right)_{i,i}}{3 M_i^2} = g_i^2$, and thus $\big({{\bf{G}}^{ 1/2}} {{\bf{U}}_\Psi }{\bf \Xi}_\Psi^2 {{\bf{U}}_\Psi } ^T  {{\bf{G}}^{ 1/2}}\big)_{i,i}  = \frac{g_i^2 {3 M_i^2}}{2 \eta^2}$, which holds for any $g_i$ by \eqref{Gcon2}. We can thus compute ${\bf \Psi}$ with ${\bf G} = {\bf I} $, obtaining
${\bf \Psi}^*= {{\bf{U}}_\Psi }{\bf \Xi}_\Psi   { \tilde{\bf \Lambda}}^{-1/2}   {\bf{ U}} ^T$, concluding the proof.
\qed

\vspace{-0.2cm}
\subsection{Proof of Theorem \ref{thm:MSE}}
 \label{app:proof2}
\vspace{-0.1cm}

Lemma~\ref{thm:MSE} characterizes the \ac{mse} minimizing sampling operator ${\bf \Psi}$ for fixed bit allocation $\{M_i\}$. To optimize $\{M_i\}$, we focus on the case where $\{M_i\}$ are not limited to be integer, resulting in the following formulation
\vspace{-0.1cm}
\begin{align}\label{OptimalProblem}
\mathop {\max }\limits_{\{ M_i \}, \{\alpha _i\}} &\sum_{i=1}^P\frac{{{\lambda _{{\bf{\Gamma }},i}^2}{\alpha _i}}}{{{\alpha _i} + 1}},\\
{\rm s.t.} &  {\left\{ {\frac{{3M_i^2}}{2\eta^2}} \right\}}  \prec  {\left\{ {{\alpha _i}} \right\}}, \,\sum_{i=1}^P \log_2 M_i \le \log_2 M, \,  M_i \ge 1. \notag
\vspace{-0.1cm}
\end{align}
Problem \eqref{OptimalProblem} is non-convex and difficult to solve, Hence, we first propose the following lemma to simplify \eqref{OptimalProblem}.
\begin{lemma}
\label{lem:opt1}
The solution of \eqref{OptimalProblem} satisfies
 ${\frac{{3M_i^2}}{2\eta^2}}  \!= \! {{\alpha _i}}$, $\forall i \in \mySet{P}$.
\end{lemma}
\begin{proof}
We first recall the following result from \cite{majorization}. 
Let $\phi: \mathcal{D}_n  \to \mathbb{R}$ be a real-valued function continuous on $\mathcal{D}_n \triangleq \{ {\bf x} \in \mathbb{R}^n : x_1 \ge ... \ge x_n  \}$ and continuously differentiable on the interior of $\mathcal{D}_n $. Then $\phi$ is Schur-convex (Schur-concave) on $\mathcal{D}_n$ if and only if $\frac{\partial\phi(\bf x)}{\partial x_i}$ is decreasing (increasing) in $i = 1,...,n$.

Assume that  $\big\{ \acute{M}_i \big\}$ is the optimal bit allocation and $\big\{ \acute{\alpha} _i \big\}$  is the corresponding diagonal elements satisfying $\{ 3\acute{M}_i^2/2\eta^2 \} \prec \{  \acute{\alpha} _i\} $. Denote $M^{(1)} = \sum_{i=1}^P \log_2 3\acute{M}_i^2/2\eta^2$ and $M^{(2)} = \sum_{i=1}^P \acute{\alpha} _i$. Since $\sum \log_2 M_i$ is Schur-concave with respect to $M_i$,  then $M^{(1)} \ge M^{(2)}$. Thus, one can propose another $\big\{ \grave{M}_i \big\}$ and $\big\{ \grave{\alpha }_i \big\} $ satisfying $3\grave{M}_i^2/2\eta^2 = \grave{\alpha }_i = \acute{\alpha} _i ^{\kappa_0}$, where $\kappa_0 = M^{(1)}/M^{(2)} \ge 1$.
Clearly, $\big\{ \grave{M}_i \big\},\big\{ \grave{\alpha }_i \big\} $ are feasible for \eqref{OptimalProblem}, with better \ac{mse}  than  $\big\{ \acute{\alpha }_i \big\} $, proving the lemma.
\end{proof}

By Proposition~\ref{lem:MSE}, the sampling matrix should be
\vspace{-0.1cm}
\begin{equation}\label{expressPsi3}
\begin{aligned}
{\bf \Psi}^*= {{\bf{U}}_\Psi }{\bf \Xi}_\Psi   { \tilde{\bf \Lambda}}^{-1/2}   {\bf{ U}} ^T
\end{aligned}
\vspace{-0.1cm}
\end{equation}
Substituting Lemma~\ref{lem:opt1} into \eqref{expressPsi3}, the matrix ${\bf{U}}_\Psi$ becomes the identity. Hence, the sampling matrix would be a diagonal matrix multiplied by the unitary  ${\bf{ U}} ^T$. The proof of Proposition~\ref{lem:MSE}  reveals that the diagonal matrix ${\bf G}^{-1/2}$ which is leftmost in the expression for ${\bf \Psi}^*$ could be arbitrarily set. As a result, the \ac{mse} is minimized by the sampling matrix ${\bf \Psi}^* = {\bf{ U}}_K ^T$.

To determine $\{M_i\}$, we write problem \eqref{OptimalProblem} as
$\mathop {\max }\limits_{\{ M_i \}} \sum_{i=1}^P \frac{{{\lambda _{{\bf{\Gamma }},i}^2}{3M_i^2}}}{{{3M_i^2} + 2\eta^2}}$, s.t.  $\sum_{i=1}^P \log_2 M_i \le \log_2 M, \,  M_i \ge 1,   \,  M_{i} \ge M_{i+1}$,
%
which is still  non-convex. To tackle this, we  obtain a local  solution and  prove its optimality.
According to Lagrange Duality theorem,
we introduce the  multipliers $\beta^* \in \mathbb{R}$  and   $\nu^* \in \mathbb{R}^P$. While this procedure does not impose monotonicity on $\{M_i\}$, it is implicitly maintained since ${\lambda _{{\bf{\Gamma }},i}}$ are arranged in a descending order. For simplicity, we denote $m_i = M_i^2$, and write the KKT constraints as $\sum_{i=1}^P \log_2 m_i^* \le  2\log_2 M$, with $m_i^* \ge 1$, $\nu_i^*(m_i^*-1) = 0$ and $\frac{{6{\lambda _{{\bf{\Gamma }},i}^2}{\eta^2}}}{\left({{3m_i} + 2\eta^2}\right)^2}-\nu_i^*-\frac{\beta^*}{m_i} = 0$, for each $ i\in\mySet{P}$. Eliminating the slack variable $\nu_i^*$, we obtain that $
\big(\frac{\beta^*}{m_i}  -\frac{{6{\lambda _{{\bf{\Gamma }},i}^2}{\eta^2}}}{\left({{3m_i} + 2\eta^2}\right)^2} \big)(m_i^*-1) = 0$ should hold.

We focus on the equation
$\frac{\beta^*}{m_i}  - \frac{{6{\lambda _{{\bf{\Gamma }},i}^2}{\eta^2}}}{\left({{3m_i} + 2\eta^2}\right)^2} = 0$,
which is solvable if and only if ${\beta^*} \le {\lambda _{{\bf{\Gamma }},i}^2}/4$.
The local optimal $m_i$ is obtained here by \eqref{algorithms1}.
The optimal bit allocation is expressed by the  dual coefficient $\beta^*$, which
  satisfies the sum bits constraint $f(\beta) = \sum_{i=1}^P \log_2 m_i^* \le  2\log_2 M$. Similar as the proof of lemma \ref{lem:opt1}, we obtain that $f(\beta) = 2\log_2 M$ should be satisfied. Further, $f(\beta)$ is a decreasing  function with respect to $\beta$, i.e., the optimal $\beta^*$ is the only one.
\qed

\vspace{-0.2cm}
\subsection{Proof of Corollary \ref{co:eq}}
 \label{app:co1}
\vspace{-0.1cm}

We first focus on the necessary condition. Due to the definition of $M_a$, $(M_a)^P := \left(\lfloor M^{1/P} \rfloor \right)^P \le M$, and $M \ge M_a^P$ should be obviously satisfied. We then assume that $M \ge M_a^P + M_a^{P-1}$, there exists another scheme of bit assignment as $[M_a + 1, M_a, \ldots, M_a]$ obtained by Algorithm~\ref{alg:greedy}, which against the previous assumption. The range of $M$ is thus determined, which corresponds to the necessity of \eqref{eqn:coeq1}.

We note that $\lambda_{\tilde \Gamma,i}$ is a descending sequence, thus  $g_i(M_i) \ge g_j(M_j)$ for $M_i = M_j$, and $g_i(M_i) > g_i(M_i+1)$ for $M_i \ge 1$.
When $M_i = M_j = M_a$ for each $i\in\mathcal{P}$, we define $k_0 = (M_a)^P-1$, as we discussed above, the state $M_i^{(k_0)}$ should be $[M_a, M_a, \ldots, M_a]$ and  $(k_0+1)$-th iteration should select $P$-th quantizer rather than others, which means that $g_P(M_a - 1) > g_i(M_a)$ for each $i\in \mathcal{P}$, and then simplified as $g_P(M_a - 1) > g_1(M_a)$. The necessity of \eqref{eqn:coeq2} is thus proven.

For the sufficient case, we assume both \eqref{eqn:coeq1} and \eqref{eqn:coeq2} are satisfied.
Greedy-based method is applied in Algorithm~\ref{alg:greedy}, we try to describe the bit assignment process.  In $k$-th iteration, the particular index is selected by $e= \arg \min_i g_i(M_i^{(k)})$.
For example, we assume that in $k_1$-th iteration, the state of $M_1^{(k_1-1)}$ change from $M_a-1$ to $M_a$, i.e., $M_1^{(k_1)} = M_a$. Here the other bit assignment should satisfy that $M_i{(k_1)} < M_a$ for each  $i=2,\ldots, P$ since $g_i(M_i) \ge g_j(M_j)$ for $M_i = M_j$.
When \eqref{eqn:coeq2} is satisfied, it holds that $g_1(M_a)  <  g_P(M_a -1) <  g_{P-1}(M_a -1) \cdots <  g_{2}(M_a - 1)$. This, before the state of $M_P$ change from $M_a-1$ to $M_a$, the state of $M_1$ would not be changed. Further,
we assume that in $k_P$-th iteration, the state of $M_P^{(k_P-1)}$ change from $M_a-1$ to $M_a$, i.e., $M_P^{(k_P)} = M_a$. Consequently, $g_{P-1}(M_a)<  g_{P-2}(M_a)  \cdots  <   g_1(M_a)  <  g_P(M_a -1)$
As a result, after the $k_P$-th iteration, $M_i  = M_a$ for each $i=1,\ldots, P$. Since \eqref{eqn:coeq1} limit the upper bound of $M$, the greedy-based algorithm would be ended here due to the termination condition, concluding the proof.
\qed

\vspace{-0.2cm}
\subsection{Proof of Corollary \ref{co:num}}
 \label{app:co2}
\vspace{-0.1cm}
We first focus on the case when sum of bits is limited, which may lead to $P < K$.
We would complete this proof in two parts, in which we separately propose that
\vspace{-0.1cm}
\begin{align}
P<i \ \ \text{if} \ \ \log_2M < l_{i} + 1, \tag{PartI}\\
P \ge i \ \  \text{if} \ \  \log_2M \ge l_{i} + 1.
\vspace{-0.1cm} \tag{PartII}
\end{align}

Define $P^*$ as the 
number of quantizers obtained by Algorithm~\ref{alg:greedy}, and its bit allocation as $\{M_1^*, M_2^*,....,M_{P^*}^*\}$. For any  $P \le P^*$, we assume that $P$-th quantizer is added in the $k_0$-th iteration, i.e., $P= \arg \max_{k_0} {\bf g}^{({k_0})}$, thus
\vspace{-0.1cm}
\begin{equation}
\begin{aligned}
g_P(1) \ge g_j(M_j^{({k_0})}) \ge  g_j(M_j^*), \forall j < P.
\vspace{-0.1cm}
\end{aligned}
\end{equation}

For the proof of the (PartI), we assume that $P\ge i$ and $\log_2M < l_{i} + 1$. Then, we define the bit allocation as $\{{\dot{M}}_1, {\dot{M}}_2,...,{\dot{M}}_P\}$, which should satisfy $g_i({\dot{M}}_i) \le g_P(1)$. It follows that ${\dot{M}}_i \ge \lceil \tilde l_j \rceil$. Then we obtain $\sum_{i = 1}^{P} {\dot{M}}_i \ge \sum_{i = 1}^{P -1} {\dot{M}}_i +1 \ge l_{i} + 1$, which contradicts the assumption.

To prove (PartII),  consider the $k_1$th iteration, where $g_i(M_i^{(k)}) \ge g_P(1)$ and $g_i(M_i^{(k)}+1) \le g_P(1)$ for $i<P$. Note that $M_i^{(k)} \le \tilde l_j \le M_i^{(k)}+1$  and $\lceil \tilde l_j \rceil = M_i^{(k)} +1$.
Thus, for the $(k_1 + P- n)$th iteration, where $1 \le n\le P-1$, $P= \arg \max_k {\bf g}^{(k)}$ is obtained, proving (PartII).
Further, when $\log_2M \ge l_K$, the $K$th quantizer is added, and there is no $(K+1)$th quantizer since the gradient vector satisfies $g_i(M_i) = 0$ for $i>K$. This results in \eqref{eqn:num}.
\qed

\vspace{-0.2cm}
\subsection{Proof of Lemma \ref{lem:Prob2x}}
 \label{app:proof4}
\vspace{-0.1cm}
To prove the equivalence between \eqref{problem2} and \eqref{problem2x}, we first note that that Lemma~\ref{lem:Digital}, which characterizes the \ac{mse} minimizing digital recovery filter, holds for any sampling matrix ${\bf \Psi}$, including those representing frequency-domain graph filtering.
Consequently, by substituting the resulting recovery matrix  ${\bf \Phi}= {\bf \Gamma}^*{\bf C}_{\bf x} {\bf \Psi}^T\left(  {\bf \Psi} {\bf C}_{\bf x} {\bf \Psi} ^T + {\bf G} \right)^{-1}$,  we arrive at
$\mathop {\max }\limits_{{\bf \Psi}, \{M_i\}} {\text {Tr}}\left({\bf \Gamma}^* {\bf C}_{\bf x} {\bf \Psi}^T \left( {\bf \Psi} {\bf C}_{\bf x} {\bf \Psi}^T + {\bf G}  \right)^{-1} {\bf \Psi} {\bf C}_{\bf x} {\bf \Gamma} ^{*T}  \right)$,  s.t. $\sum_{i=1}^P \log_2 M_i \le \log_2M, M_i \in \mathbb{Z}^+$.
By setting  ${\bf \Psi} = {\bf I}_\mathcal{S} {\bf U} F({\bf \Lambda}) {\bf U}^T $ to represent frequency-domain filtering, we transform the objective into  \eqref{problem2x},  proving the lemma.
\qed

\vspace{-0.2cm}
\subsection{Proof of Lemma \ref{lem:Prob5x}}
 \label{app:proof5}
\vspace{-0.1cm}

We first transform $({ {\tilde\lambda _i^2{{\bf{B}}} + {\bf{C}}} })^{-1}$ into
\vspace{-0.1cm}
\begin{align}
 ({ {\tilde\lambda _i^2{{\bf{B}}} + {\bf{C}}} })^{-1} & = {{{\bf{B}}}^{ - 1/2}} {\left( {\tilde\lambda _i^2{\bf{I}} + {{{\bf{B}}}^{ - 1/2}}{{\bf{C}}}{{{\bf{B}}}^{ - 1/2}}} \right)^{ - 1}{{{\bf{B}}}^{ - 1/2}}} \notag \\
& \overset{(a)}{=}  {{{\bf{B}}}^{ - 1/2}}{\left( \tilde\lambda _i^2{\bf{I}} + {\bf U}_x {\bf{\Lambda }}_x {\bf U}_x^T \right)^{ - 1}{{{\bf{B}}}^{ - 1/2}}} \notag \\
& = {{{\bf{B}}}^{ - 1/2}}{\bf U}_x{\left( \tilde\lambda _i^2{\bf{I}} +  {\bf{\Lambda }}_x \right)^{ - 1}{\bf U}_x^T {{{\bf{B}}}^{ - 1/2}}},
\vspace{-0.1cm}
\label{Decomp1}
\end{align}
where $(a)$ follows from the eigenvalue  decomposition, where ${\bf U}_x$ is a unitary matrix since both ${\bf{B}}$ and ${\bf{C}}$ are symmetric.
Substituting \eqref{Decomp1} and  the definition of ${\bf A}$ into \eqref{problem5x} yields
\vspace{-0.1cm}
\begin{equation}\label{problem5p1}
\begin{aligned}
\mathop {\max }\limits_{\tilde \lambda_i} {\rm{Tr}}\left( {\bf A}{\tilde {\bf{\Lambda }}}^2 F{({\bf{\Lambda }})}^2  {\bf A}^T \left( \tilde\lambda _i^2{\bf{I}} +  {\bf{\Lambda }}_x \right)^{ - 1}  \right).
\end{aligned}
\vspace{-0.1cm}
\end{equation}

Noting that $\tilde \lambda_i$ is in form of $\tilde \lambda_i^2$ in \eqref{problem5p1}, thus we define $\hat \lambda_i = \tilde \lambda_i^2$, simplifying \eqref{problem5p1} into:
$\mathop {\max }\limits_{\hat \lambda_i \geq 0} ~ \sum_{j=1}^P \frac{\sum_{n=1}^N\hat \lambda_n ({\bf A})_{j,n}^2   }{\hat \lambda_i + {({\bf{\Lambda }}_x)}_{j,j}}$.
Substituting the definition of $\{a_i\}$
yields   \eqref{ComputeLambda}. 
\qed

  \end{appendix}

	\bibliographystyle{IEEEtran}
	\bibliography{IEEEabrv,refs}

\begin{thebibliography}{10}
\providecommand{\url}[1]{#1}
\csname url@samestyle\endcsname
\providecommand{\newblock}{\relax}
\providecommand{\bibinfo}[2]{#2}
\providecommand{\BIBentrySTDinterwordspacing}{\spaceskip=0pt\relax}
\providecommand{\BIBentryALTinterwordstretchfactor}{4}
\providecommand{\BIBentryALTinterwordspacing}{\spaceskip=\fontdimen2\font plus
\BIBentryALTinterwordstretchfactor\fontdimen3\font minus
  \fontdimen4\font\relax}
\providecommand{\BIBforeignlanguage}[2]{{%
\expandafter\ifx\csname l@#1\endcsname\relax
\typeout{** WARNING: IEEEtran.bst: No hyphenation pattern has been}%
\typeout{** loaded for the language `#1'. Using the pattern for}%
\typeout{** the default language instead.}%
\else
\language=\csname l@#1\endcsname
\fi
#2}}
\providecommand{\BIBdecl}{\relax}
\BIBdecl

\bibitem{li2021graph}
P.~Li, N.~Shlezinger, H.~Zhang, B.~Wang, and Y.~C. Eldar, ``Graph signal
  compression via task-based quantization,'' in \emph{Proc. IEEE ICASSP}, 2021.

\bibitem{ref1}
D.~I. {Shuman}, S.~K. {Narang}, P.~{Frossard}, A.~{Ortega}, and
  P.~{Vandergheynst}, ``The emerging field of signal processing on graphs:
  Extending high-dimensional data analysis to networks and other irregular
  domains,'' \emph{{IEEE} Signal Process. Mag.}, vol.~30, no.~3, pp. 83--98,
  2013.

\bibitem{ref2}
A.~{Sandryhaila} and J.~M.~F. {Moura}, ``Big data analysis with signal
  processing on graphs: Representation and processing of massive data sets with
  irregular structure,'' \emph{{IEEE} Signal Process. Mag.}, vol.~31, no.~5,
  pp. 80--90, 2014.

\bibitem{ref3}
------, ``Discrete signal processing on graphs,'' \emph{{IEEE} Trans. Signal
  Process.}, vol.~61, no.~7, pp. 1644--1656, 2013.

\bibitem{GSPsurvey}
A.~{Ortega}, P.~{Frossard}, J.~{Kovačević}, J.~M.~F. {Moura}, and
  P.~{Vandergheynst}, ``Graph signal processing: Overview, challenges, and
  applications,'' \emph{Proc. {IEEE}}, vol. 106, no.~5, pp. 808--828, 2018.

\bibitem{SamplingTheory}
S.~{Chen}, R.~{Varma}, A.~{Sandryhaila}, and J.~{Kovačević}, ``Discrete
  signal processing on graphs: Sampling theory,'' \emph{{IEEE} Trans. Signal
  Process.}, vol.~63, no.~24, pp. 6510--6523, 2015.

\bibitem{tanaka2020generalized}
Y.~Tanaka and Y.~C. Eldar, ``Generalized sampling on graphs with subspace and
  smoothness priors,'' \emph{{IEEE} Trans. Signal Process.}, vol.~68, pp.
  2272--2286, 2020.

\bibitem{tanaka2020sampling}
Y.~Tanaka, Y.~C. Eldar, A.~Ortega, and G.~Cheung, ``Sampling on graphs: From
  theory to applications,'' \emph{{IEEE} Signal Process. Mag.}, vol.~37, no.~6,
  pp. 14--30, 2020.

\bibitem{FrequencyAnalysis}
A.~{Sandryhaila} and J.~M.~F. {Moura}, ``Discrete signal processing on graphs:
  Frequency analysis,'' \emph{{IEEE} Trans. Signal Process.}, vol.~62, no.~12,
  pp. 3042--3054, 2014.

\bibitem{GSPdenoising}
A.~C. {Yağan} and M.~T. {Özgen}, ``Spectral graph based vertex-frequency
  wiener filtering for image and graph signal denoising,'' vol.~6, pp.
  226--240, 2020.

\bibitem{heimowitz2018markov}
A.~Heimowitz and Y.~C. Eldar, ``A {Markov} variation approach to smooth graph
  signal interpolation,'' \emph{arXiv preprint arXiv:1806.03174}, 2018.

\bibitem{GSPsmooth}
X.~{Dong}, D.~{Thanou}, P.~{Frossard}, and P.~{Vandergheynst}, ``Learning
  laplacian matrix in smooth graph signal representations,'' \emph{{IEEE}
  Trans. Signal Process.}, vol.~64, no.~23, pp. 6160--6173, 2016.

\bibitem{GSPsparse}
S.~P. {Chepuri}, S.~{Liu}, G.~{Leus}, and A.~O. {Hero}, ``Learning sparse
  graphs under smoothness prior,'' in \emph{Proc. IEEE ICASSP}, 2017.

\bibitem{SamplingGreedy}
L.~F.~O. {Chamon} and A.~{Ribeiro}, ``Greedy sampling of graph signals,''
  \emph{{IEEE} Trans. Signal Process.}, vol.~66, no.~1, pp. 34--47, 2018.

\bibitem{gray1998quantization}
R.~M. Gray and D.~L. Neuhoff, ``Quantization,'' \emph{{IEEE} Trans. Inf.
  Theory}, vol.~44, no.~6, pp. 2325--2383, 1998.

\bibitem{QuantizationFilter}
L.~F.~O. {Chamon} and A.~{Ribeiro}, ``Finite-precision effects on graph
  filters,'' in \emph{Proc. IEEE GlobalSIP}, 2017, pp. 603--607.

\bibitem{QuantizationDistributed}
I.~C.~M. {Nobre} and P.~{Frossard}, ``Optimized quantization in distributed
  graph signal processing,'' in \emph{Proc. IEEE ICASSP}, 2019.

\bibitem{QuantizationTimeVarying}
L.~B. {Saad}, E.~{Isufi}, and B.~{Beferull-Lozano}, ``Graph filtering with
  quantization over random time-varying graphs,'' in \emph{Proc. IEEE
  GlobalSIP}, 2019.

\bibitem{QuantizationInterpolation}
P.~{Di Lorenzo}, S.~{Barbarossa}, and P.~{Banelli}, ``Optimal power and bit
  allocation for graph signal interpolation,'' in \emph{Proc. IEEE ICASSP},
  2018.

\bibitem{QuantizationSampling}
Y.~H. {Kim} and A.~{Ortega}, ``Toward optimal rate allocation to sampling sets
  for bandlimited graph signals,'' \emph{{IEEE} Signal Process. Lett.},
  vol.~26, no.~9, pp. 1364--1368, 2019.

\bibitem{HardwareLimited}
N.~{Shlezinger}, Y.~C. {Eldar}, and M.~R.~D. {Rodrigues}, ``Hardware-limited
  task-based quantization,'' \emph{{IEEE} Trans. Signal Process.}, vol.~67,
  no.~20, pp. 5223--5238, 2019.

\bibitem{shlezinger2018asymptotic}
N.~Shlezinger, Y.~C. Eldar, and M.~R. Rodrigues, ``Asymptotic task-based
  quantization with application to massive {MIMO},'' \emph{{IEEE} Trans. Signal
  Process.}, vol.~67, no.~15, pp. 3995--4012, 2019.

\bibitem{Salamtian19task}
S.~Salamtian, N.~Shlezinger, Y.~C. Eldar, and M.~Medard, ``Task-based
  quantization for recovering quadratic functions using principal inertia
  components,'' in \emph{Proc. IEEE ISIT}, 2019.

\bibitem{neuhaus2020taskbased}
P.~Neuhaus, N.~Shlezinger, M.~Dörpinghaus, Y.~C. Eldar, and G.~Fettweis,
  ``Task-based analog-to-digital converters,'' \emph{{IEEE} Trans. Signal
  Process.}, early access, 2021.

\bibitem{shlezinger2020taskbased}
N.~Shlezinger and Y.~C. Eldar, ``Task-based quantization with application to
  {MIMO} receivers,'' \emph{arXiv preprint arXiv:2002.04290}, 2020.

\bibitem{shlezinger2019deep}
------, ``Deep task-based quantization,'' \emph{Entropy}, vol.~23, no.~1, p.
  104, 2021.

\bibitem{gama2020graphs}
F.~Gama, E.~Isufi, G.~Leus, and A.~Ribeiro, ``Graphs, convolutions, and neural
  networks: From graph filters to graph neural networks,'' \emph{{IEEE} Signal
  Process. Mag.}, vol.~37, no.~6, pp. 128--138, 2020.

\bibitem{GraphModel1}
X.~{Dong}, D.~{Thanou}, P.~{Frossard}, and P.~{Vandergheynst}, ``Learning
  laplacian matrix in smooth graph signal representations,'' \emph{{IEEE}
  Trans. Signal Process.}, vol.~64, no.~23, pp. 6160--6173, 2016.

\bibitem{2011Wavelets}
D.~K. Hammond, P.~Vandergheynst, and R.~Gribonval, ``Wavelets on graphs via
  spectral graph theory,'' \emph{Applied and Computational Harmonic Analysis},
  vol.~30, no.~2, pp. 129--150, 2011.

\bibitem{gray1993dithered}
R.~M. Gray and T.~G. Stockham, ``Dithered quantizers,'' \emph{{IEEE} Trans.
  Inf. Theory}, vol.~39, no.~3, pp. 805--812, 1993.

\bibitem{widrow1996statistical}
B.~Widrow, I.~Kollar, and M.-C. Liu, ``Statistical theory of quantization,''
  \emph{{IEEE} Trans. Instrum. Meas.}, vol.~45, no.~2, pp. 353--361, 1996.

\bibitem{polyanskiy2014lecture}
Y.~Polyanskiy and Y.~Wu, ``Lecture notes on information theory,'' \emph{Lecture
  Notes for 6.441 (MIT), ECE563 (University of Illinois Urbana-Champaign), and
  STAT 664 (Yale)}, 2012-2017.

\bibitem{palomar2007mimo}
D.~P. Palomar and Y.~Jiang, \emph{{MIMO} transceiver design via majorization
  theory}.\hskip 1em plus 0.5em minus 0.4em\relax Now Publishers Inc, 2007.

\bibitem{quadratic2018}
K.~Shen and W.~Yu, ``Fractional programming for communication systems—part
  {I}: Power control and beamforming,'' \emph{{IEEE} Trans. Signal Process.},
  vol.~66, no.~10, pp. 2616--2630, 2018.

\bibitem{cover2012elements}
T.~M. Cover and J.~A. Thomas, \emph{Elements of Information Theory}.\hskip 1em
  plus 0.5em minus 0.4em\relax John Wiley \& Sons, 2012.

\bibitem{shlezinger2020uveqfed}
N.~Shlezinger, M.~Chen, Y.~C. Eldar, H.~V. Poor, and S.~Cui, ``{UVeQFed}:
  Universal vector quantization for federated learning,'' \emph{{IEEE} Trans.
  Signal Process.}, vol.~69, pp. 500--514, 2021.

\bibitem{GSPBOX}
N.~Perraudin, J.~Paratte, D.~Shuman, L.~Martin, V.~Kalofolias,
  P.~Vandergheynst, and D.~K. Hammond, ``{GSPBOX}: A toolbox for signal
  processing on graphs,'' \emph{arXiv preprint arXiv:1408.5781}, 2014.

\bibitem{refimage1}
W.~Hu, G.~Cheung, A.~Ortega, and O.~C. Au, ``Multiresolution graph fourier
  transform for compression of piecewise smooth images,'' \emph{IEEE
  Transactions on Image Processing}, vol.~24, no.~1, pp. 419--433, 2015.

\bibitem{majorization}
A.~W. Marshall, I.~Olkin, and B.~C. Arnold, \emph{Inequalities: theory of
  majorization and its applications}.\hskip 1em plus 0.5em minus 0.4em\relax
  Springer, 1979, vol. 143.

\end{thebibliography}

\end{document}